\documentclass[11pt]{amsart}
\usepackage{hyperref}

\usepackage{amsfonts,amsmath,amssymb}
\usepackage[margin=1in]{geometry}

\usepackage{algorithm}
\usepackage[noend]{algpseudocode}

\usepackage{graphicx}
\newtheorem{theorem}{Theorem}[section]
\newtheorem{proposition}[theorem]{Proposition}
\newtheorem{lemma}[theorem]{Lemma}

\newtheorem{definition}[theorem]{Definition}

\newcommand{\bE}{\ensuremath{\mathbf{E}}}

\begin{document}

\title[Deterministic parallel algorithms for bilinear objectives]{Deterministic parallel algorithms for bilinear objective functions}

\author[David G. Harris]{
{\sc David G.~Harris}$^{1}$
}

\setcounter{footnote}{0}
\addtocounter{footnote}{1}
\footnotetext{Department of Computer Science, University of Maryland, 
College Park, MD 20742. 
Research supported in part by NSF Awards CNS-1010789 and CCF-1422569.
Email: \texttt{davidgharris29@gmail.com}.}

\maketitle

\begin{abstract}
Many randomized algorithms can be derandomized efficiently using either the method of conditional expectations or probability spaces with low independence. A series of papers, beginning with work by Luby (1988), showed that in many cases these techniques can be combined to give deterministic parallel (NC) algorithms for a variety of combinatorial optimization problems, with low time- and processor-complexity.

We extend and generalize a technique of Luby for efficiently handling bilinear objective functions. One noteworthy application is an NC algorithm for maximal independent set. On a graph $G$ with $m$ edges and $n$ vertices, this takes $\tilde O(\log^2 n)$  time and $(m + n) n^{o(1)}$ processors, nearly matching the best randomized parallel algorithms. Other applications include reduced processor counts for algorithms of Berger (1997) for maximum acyclic subgraph and Gale-Berlekamp switching games.

This bilinear factorization also gives better algorithms for problems involving discrepancy. An important application of this is to automata-fooling probability spaces, which are the basis of a notable derandomization technique of Sivakumar (2002). Our method leads to large reduction in processor complexity for a number of derandomization algorithms based on automata-fooling, including set discrepancy and the Johnson-Lindenstrauss Lemma.
\end{abstract}

\section{Introduction}
Let us consider the following scenario, which frequently appears in applications of the probabilistic method of combinatorics. We have some objective function $S(x)$ we wish to maximize over $x \in \{0, 1 \}^n$. Furthermore, we know that, if the random variable $X$ is drawn from a $w$-wise-independent probability space over $\{0, 1 \}^n$, then $\bE[S(X)] \geq T$. Then there certainly exists $x \in  \{0, 1 \}^n$ satisfying $S(x) \geq T$. A key challenge for deterministic algorithms is to find such $x$ efficiently and in parallel.

As $w$-wise-probability spaces have size roughly $n^w$, we could exhaustively enumerate over such a space using $O(n^w)$ processors. Alternatively, Luby \cite{luby-old} noted that a binary search can be used; this takes advantage of the fact that $w$-wise-independent probability spaces come from binary linear codes. Variants of this method are used by \cite{berger-simulating, harris} for NC derandomizations of a number of algorithms, including set discrepancy, rainbow hypergraph coloring and the Lov\'{a}sz Local Lemma.

The main cost during this binary search is the evaluation of $S(X)$, given that $X$ is confined to some lower-dimensional subspace of $\{0, 1 \}^n$. This leads to high processor complexities, as typically one needs a single processor for each summand of $S$.  When the objective function has a certain nice ``bilinear'' form, then Luby \cite{luby-old} noted that binary search can be applied without ever representing $S$ explicitly. A simple example would be
$$
S(x) = \sum_{\substack{(\gamma_1, e_1) \in U_1\\(\gamma_2, e_2) \in U_2}} \gamma_1 \gamma_2 (-1)^{x \bullet (e_1 \oplus e_2)}
$$
where $U_1, U_2$ are subsets of $\mathbf R \times 2^{[n]}$, and where $\oplus$ represents XOR and $\bullet$ is the mod-2 inner product. Instead of evaluating this function by unrolling it into $|U_1| |U_2|$ separate summands, one can compute the necessary conditional expectations directly in this bilinear form, leading to a significant reduction in processor count. 

This type of bilinear objective function is surprisingly common and the technique of \cite{luby-old} is quite powerful, but it seems to have fallen by the wayside in more recent derandomization research. With a few exceptions (e.g. \cite{berger-setcover, low-diam}), these techniques have not been applied as widely as they could be. 

\subsection{Overview of our approach and improvements}
Our goal in this paper is to extend this technique to cover a greater range of algorithmic applications than considered by \cite{luby-old}, including general classes of problems involving discrepancy and finite-state automata.

In Section~\ref{bilinear-sec}, we develop a general framework to handle different types of objective functions. In particular, we overcome a major technical limitation of \cite{luby-old} by handling non-binary random variables and non-linear objective functions in an efficient and clean way.

In Section~\ref{mis-sec}, we apply our framework to maximal independent set (MIS) of a graph $G = (V,E)$. The key innovation here is our ability to handle non-binary random variables. This allows us to simulate Luby's randomized MIS algorithm \cite{luby-mis}, which uses a relatively complex probability distribution. We obtain an essentially optimal derandomization of Luby's MIS algorithm \cite{luby-mis}, with time complexity approximately $O(\log^2 |V|)$ and processor complexity approximately $O(|E|+|V|)$. 

In Section~\ref{Bffourth-moment-sec}, we derandomize  the Gale-Berlekamp Switching Game and the maximum acyclic subgraph problem, which are two problems from a class of algorithms based on Berger's fourth-moment method \cite{berger}.

In Section~\ref{symmetric-sec}, we consider discrepancy minimization. In this setting, there are $m$ linear functionals in $n$ variables, and we wish to ensure that all linear functionals are simultaneously close to their means. When the variables are independent, then very strong tail inequalities, such as Chernoff's bound, apply; when the variables are selected with $w$-wise-independence, then weaker bounds such as Chebyshev's inequality apply instead.  We show that $w$-wise independence can be simulated using $(mn)^{o(1)} m n^{\lceil w/2 \rceil}$ processors and $\tilde O(\log n)$ time; by contrast, previous work such as \cite{motwani-naor-naor} would use roughly $m n^w$ processors.

Discrepancy minimization plays a ubiquitous role in algorithm design. In Section~\ref{Bfsec:fool}, we consider one particularly powerful application to fooling finite-state automata developed by Nisan \cite{nis1, nis2}. Sivakumar \cite{sivakumar} used this to derandomize algorithms based on low-memory statistical tests. The original processor complexities for these algorithms, while polynomial, were extremely high. In \cite{mrs}, Mahajan et al. optimized these for certain types of automata based on counters. We optimize these further, reducing the processor count significantly and covering more general classes of automata.

In Section~\ref{app-fool}, we apply these results to two fundamental problems for which randomized algorithms can give very good results: set discrepancy, and the Johnson-Lindenstrauss Lemma. Here, the approach of \cite{sivakumar} gives simple and clean derandomizations, albeit with very high processor counts. We obtain much lower processor counts for these problems. For example, for set discrepancy with $m$ linear functionals on $n$ variables, we require $\tilde O(\log^2 n)$ time and $m^{4+o(1)} n^{3+o(1)}$ processors.

\subsection{Model of computation}
This paper focuses on the deterministic PRAM model, in which there is a common memory, $\text{poly}(n)$ processors and $\text{polylog}(n)$ time on an input of length $n$. We focus on the most stringent variant, the EREW PRAM, in which no simultaneous access (either reading or writing) is allowed to a single memory cell. Other models, for instance the CRCW PRAM model, allow joint access; we can simulate a step of CRCW with a $\log n$ overhead in time and processor complexity.

We say an algorithm has \emph{complexity} $(C_1, C_2)$ if it uses $O(C_1)$ time and $O(C_2)$ processors on a deterministic EREW PRAM. Our goal in this paper is to optimize algorithm complexity, and we often wish to focus on the first-order terms; for this reason, we say an algorithm has \emph{quasi-complexity} $(C_1, C_2)$ if it has complexity $( C_1 \text{polyloglog}(n), C_2 n^{o(1)} )$; that is, we ignore $\text{polyloglog}$ terms in the time and sub-polynomial terms in processor count. We note that many simple operations (such as adding integers) runs in $\text{polyloglog}(n)$ time, depending on more precise details of the PRAM model (for example, the size of the memory cells); thus, giving complexity bounds which are finer than quasi-complexity may be very difficult.

\subsection{Notation}
We let $[t]$ throughout denote the set $[t] = \{1, \dots, t \}$. For a probability space $\Omega$, we say that $x \sim \Omega$ if $x$ is a random variable drawn from $\Omega$; we define the \emph{size} of $\Omega$ as the cardinality of its support. For a finite set $X$, we write $x \sim X$ if $x$ is drawn from the uniform distribution on $X$.

For a boolean predicate $\mathcal P$, we use the Iverson notation $[ \mathcal P ]$  for the indicator function which is 1 if $\mathcal P$ is true, and 0 otherwise.

\section{Bilinearizable conditional expectations}
\label{bilinear-sec}
In this section, we develop NC algorithms for various types of objective functions. We begin with a particularly simple type of linear objective function (analyzed by \cite{luby-old} using different terminology).  We define an \emph{ensemble} $\mathcal E$ to be a pair of functions $\mathcal E^1, \mathcal E^2: 2^{[n]} \rightarrow \mathbf R$. We typically assume that these are highly sparse, and are represented by an explicit listing of their (few) non-zero entries. We define $\langle \mathcal E \rangle$ to be the support of $\mathcal E$, and we define $|\mathcal E |$ to be the total size needed to store $\mathcal E$, i.e. $|\mathcal E | = \sum_{e \in \langle E \rangle} |e|$. For any $e \subseteq [n]$ and vector $x \in \{0,1 \}^n$, we define $x \bullet e = \sum_{i \in e} x_i$. 

For an ensemble $\mathcal E$, we define the objective function $S_{\mathcal E}: \{0, 1 \}^n \rightarrow \mathbf R$ by 
$$
S_{\mathcal E}(x) = \sum_{e_1, e_2 \subseteq [n]} \mathcal E^1(e_1) \mathcal E^2(e_2) (-1)^{x \bullet(e_1 \oplus e_2)}
$$
where  $ \oplus$ is the coordinate-wise XOR (or symmetric difference) of the sets.  We also define
$$
T(\mathcal E) := \sum_{ e \subseteq [n] } \mathcal E^1(e) \mathcal E^2(e)
$$

It is straightforward to verify that $\bE[S_{\mathcal E}(X)] = T(\mathcal E)$ for $X \sim \{0, 1\}^n$. This is our ``benchmark'' result, and we will try to constructively find some $x \in \{0, 1 \}^n$ with $S_{\mathcal E}(x) \geq T(\mathcal E)$. Instead of directly searching the solution space (for example, searching a $w$-wise-independent probability space), we take advantage of the fact that low-independence spaces can be represented as binary codes.

We use a somewhat nonstandard terminology in our discussion of such codes. A binary code is defined by associating an $L$-long binary vector $A(i) \in \{0,1 \}^L$ to each $i \in [n]$; we refer to $L$ as the \emph{length} of the code.  For any $e \subseteq [n]$ we define $A(e) \in \{0, 1\}^L$ as the coordinate-wise XOR given by $A(e) := \bigoplus_{i \in e}  A(i)$.

\begin{definition}
Given an ensemble $\mathcal E$ and a code $A$ over $\{0, 1 \}^n$, we say that $A$ \emph{fools} $\mathcal E$ if
$$
A(e_1) \neq A(e_2) \qquad \text{for all distinct $e_1, e_2$ in $\langle E \rangle$}
$$
\end{definition}

\begin{lemma}
  \label{bilinear-ce-prop}
  Let $\mathcal E_1, \dots, \mathcal E_m$ be ensembles on $n$ variables and let $A$ be a code of length $L = \Theta(\log n)$ which fools $\mathcal E_1, \dots, \mathcal E_m$. Then in quasi-complexity $(\log n \sum_j | \mathcal E_j |,  \sum_j |\mathcal E_j |)$ we can find $x \in \{0, 1 \}^n$ satisfying $\sum_{j=1}^m S_{\mathcal E_j}(x) \geq \sum_{j=1}^m T(\mathcal E_j)$.
\end{lemma}
\begin{proof}Consider the function $G: \{0, 1 \}^L \rightarrow \mathbf R$ by
$$
G(y) = \sum_{j=1}^m \sum_{e_1, e_2 \subseteq [n]} \mathcal E_j^1(e_1) \mathcal E_j^2(e_2) (-1)^{y \bullet A(e_1 \oplus e_2)}
$$

We will construct a vector $y \in \{0,1 \}^L$ with $G(y) \geq \sum_j T(\mathcal E_j)$. We will then take our solution vector $x$ defined by $x_i = A(i) \bullet y$; this will satisfy
\begin{align*}
\sum_j S_{\mathcal E_j}(x) &= \sum_j \sum_{\substack{e_1 \subseteq [n] \\ e_2 \subseteq [n]}} \mathcal E^1_j(e_1) \mathcal E^2_j(e_2) (-1)^{x \bullet (e_1 \oplus e_2)} = \sum_j \sum_{\substack{e_1 \subseteq [n] \\ e_2 \subseteq [n]}} \mathcal E^1_j(e_1) \mathcal E^2_j(e_2) (-1)^{y \bullet A(e_1 \oplus e_2)} = G(y)
\end{align*}
and so the vector $x \in \{0,1\}^n$ has $\sum_{j=1}^m S_{\mathcal E_j}(x) = G(y) \geq \sum_{j=1}^m T(\mathcal E_j)$ as required.

To construct $y$, we will proceed through $L/t$ stages for parameter $t = \frac{\log n}{\log \log n}$,  For each such stage $\ell = 0, \dots, L/t$, we will determine a vector $y^{\ell} \in \{0,1 \}^{\ell t}$ such that
\begin{equation}
\label{rr2r}
\bE_{Y \sim \{0,1\}^L}[G(Y) \mid (Y_1, \dots, Y_{\ell t}) = y^{\ell}] \geq \sum_j T(\mathcal E_j)
\end{equation}

To begin, we set $y^0$ to be the empty vector. To verify that this works, observe that
\begin{align*}
\bE[ G(Y) ] &= \sum_{j=1}^m \sum_{\substack{e_1, e_2 \subseteq [n]}} \mathcal E_j^1(e_1) \mathcal E_j^2(e_2) \bE[(-1)^{y \bullet A(e_1 \oplus e_2)}] = \sum_{j=1}^m \sum_{\substack{e_1,  e_2 \subseteq[n] \\ A(e_1) = A(e_2)}} \mathcal E_j^1(e_1) \mathcal E_j^2(e_2) \\
&= \sum_{j=1}^m \sum_{\substack{e_1, e_2 \subseteq[n] \\ e_1 = e_2}} \mathcal E_j^1(e_1) \mathcal E_j^2(e_2)  \qquad \text{since $A$ fools $\mathcal E_j$} \\
&= \sum_{j=1}^m \sum_{e \subseteq [n]} \mathcal E_j^1(e) \mathcal E_j^2(e) = \sum_{j=1}^m T(\mathcal E_j)
\end{align*}

In each stage $\ell = 0, \dots, L/t$, we search over all $2^t = n^{o(1)}$ possible extensions of $y^{\ell}$ to $y^{\ell+1}$. We compute the conditional expected value of $G(Y)$, and select the choice of $y^{\ell}$ to maximize $\bE[G(Y)]$. This ensures that the conditional expected value increases at each stage, and so (\ref{rr2r}) holds. At the end setting $y = y^{L/t}$ achieves $G(y) \geq \sum_j T(\mathcal E_j)$.

In order to carry out this process, we must compute $\bE[G(Y) \mid (Y_1, \dots, Y_{\ell t}) = q]$ for a given vector $q \in \{0,1\}^{\ell t}$; this is the main computational cost of the algorithm. To do so, note that each summand $(-1)^{y \oplus A(e_1 \oplus e_2)}$ has mean zero if there is a coordinate $i \in \{\ell t + 1, \dots, L \}$ such that $A(e_1 \oplus e_2)(i) = 1$, and otherwise $(-1)^{y \oplus A(e_1 \oplus e_2)}$ is completely determined from the value $q$. Let us define $\pi: \{0, 1 \}^L \rightarrow \{0,1 \}^{L - q}$ to be the projection onto the final $L-q$ coordinates; thus we can write
$$
\bE[ G(y) \mid (Y_1, \dots, Y_{\ell t}) = q ] = \sum_j \sum_{\substack{ e_1, e_2 \subseteq [n] \\ \pi(A(e_1)) = \pi(A(e_2))}} \mathcal E^1_j(e_1) \mathcal E^2_j(e_2) (-1)^{q \bullet (A(e_1) \oplus A(e_2))}
$$

By sorting $\langle \mathcal E_j \rangle$, this quantity can be computed using quasi-complexity $(\log (n \sum_j |\mathcal E_j|), \sum_j |\mathcal E_j|)$. Since there are $L/t$ stages, and at each stage we must search for $2^t$ possible values for the value $y^{\ell+1}$, the overall processor count is $2^t \times  \sum_j |\mathcal E_j|$ and the overall runtime is $L/t \times \log(n \sum_j |\mathcal E_j|)$. With our choice of parameter $t = \frac{\log n}{\log \log n}$ and $L = \Theta(\log n)$, this gives an overall quasi-complexity of $(\log(n \sum_j |\mathcal E_j|), \sum_j |\mathcal E_j|)$.
\end{proof}

This is essentially optimal complexity, as it would require that much time and space to simply store the ensembles $\mathcal E_1, \dots, \mathcal E_m$.

\subsection{Bilinearizable objective functions}
\label{bilinearizable-sec}
We now consider more general scenarios, in which the variables are not single bits but take values in the larger space $\mathcal B_b = \{0, 1 \}^b$, and where the objective functions are complex functions of these variables.  As a starting point, we are interested in objective functions of the form
\begin{equation}
\label{bilinear-gen1}
S(x) = \sum_{k_1 = 1}^{m_1} \sum_{k_2 = 1}^{m_2} f_{1,k_1}(x) f_{2,k_2}(x)
\end{equation}
where each sub-function $f_{\ell,k}$ depends on $w$ coordinates of $x$ (such a function is referred to as a \emph{$w$-junta}.)

We view the space $\mathcal B_b$ as consisting of $b$ separate ``bit-levels''. There is an obvious identification between $\mathcal B_b$ and the integer set $\{0, 1, 2, \dots, 2^b - 1 \}$; so, we can identify the bit-levels from least-significant (rightmost) bits to most-significant (leftmost) bits. For any $v \in \mathcal B_b$, we define $\langle i:j \rangle$ to the subvector consisting of bits $i, \dots, j$ of $v$. Thus $v \langle 1:j \rangle$ represents the most-significant $j$ bits of the integer $v$ and $v \langle b-j+1: b \rangle$ are the least-significant $j$ bits of the integer $v$. We also write $x \# y$ to mean the concatenation of the two bit-strings $x,y$.

For a vector $x \in \mathcal B_b^n$, we define $x \langle i:j \rangle$ to be the vector $(x_1 \langle i:j \rangle, \dots, x_n \langle i:j \rangle)$. Likewise, for vectors $x \in \mathcal B_s^n, y \in \mathcal B_t^n$ we define $x \# y = (x_1 \# y_1, \dots, x_n \# y_n)$.

The basic strategy of our conditional expectations algorithm is to fix the bit-levels of the vector $x$ sequentially. For each bit-level in turn, we use a Fourier transform to transform (\ref{bilinear-gen1}) to a bilinear objective function. The key requirement for this is a certain type of factorization of the objective function $S(x)$.
\begin{definition}
  \label{def-bilin}
We say that an objective function $S: \mathcal B_b^n \rightarrow \mathbf R$ has  \emph{bilinear expectations  with width $w$ on window $t$}, if for every $b' \leq b$ there is an (explicitly given) family of functions $F$ such that for all $f^0 \in \mathcal B_{b'}^n, f^1 \in \mathcal B_t^n$ we have
\begin{align*}
\bE_{X \sim \mathcal B_b^n} \Bigl[ S(X) \mid X \langle 1:b' + t \rangle = f^0 \# f^1 \Bigr] = \sum_{j=1}^m \sum_{k_1 = 1}^{m_{j,1}} \sum_{k_2 = 1}^{m_{j,2}} F_{j,1,k_1}(f^0, f^1) F_{j,2,k_2}(f^0, f^1)
\end{align*}
and  each function $F_{j,\ell,k}$ depends on only $w$ coordinates of $f^1$. (It may depend arbitrarily on $f^0$).

We define the \emph{weight} of this factorization as $W = \sum_{j=1}^m (m_{j_1} + m_{j_2})$.

Frequently, when we use write a function $F$ in this form, we omit the dependence on $f^0$, which we view as fixed and arbitrary. The critical part of this definition is how $F$ depends on $f^1$. Specifically, for fixed $f^0$, the function $F$ must be a $w$-junta.
\end{definition}

\begin{theorem}
\label{thm-bilinear-gen}
Suppose that an objective function $S: \mathcal B_b^n \rightarrow \mathbf R$ has a bilinear-expectations factorization of weight $W$,  window $t = O( \frac{\log n}{\log \log n})$ and width $w = O(1)$, and $b = O(\log n)$, and this factorization can be determined explicitly with complexity $(C_1, C_2)$. Then there is an algorithm to find $x \in \mathcal B_b^n$ with $S(x) \geq \bE_{X \sim \mathcal B_b^n}[S(X)]$ using quasi-complexity $(C_1 + \log(W n), C_2 + W + n)$.
\end{theorem}
\begin{proof}
We gradually fix the bit-levels of $X$ in chunks of $t$, starting with the most-significant bit-levels, for $b/t$ separate iterations. Specifically, at each step $\ell = 1, \dots, b/t$, we find a vector $y^\ell \in \mathcal B_{\ell t}^n$ such that
$$
\bE_{X \sim \mathcal B_b^n}[S(X) \mid X \langle 1:\ell t \rangle = y^\ell] \geq \bE_{X \sim \mathcal B_b^n}[S(X)]
$$

This holds vacuously for $y^0$ being the empty vector. To go from step $\ell$ to $\ell+1$, we use the bilinear expectations with $b' = \ell t$ to write
$$
G(z) = \bE[S(X) \mid X \langle 1: (\ell+1) t \rangle = y^\ell \# z ] = \sum_{j=1}^m \sum_{k_1 = 1}^{m_{j,1}} \sum_{k_2 = 1}^{m_{j,2}} F_{j,1,k_1}(y^{\ell}, z) F_{j,2,k_2}(y^{\ell}, z)
$$
for any $z \in \mathcal B_t^n$. Having fixed $y^{\ell}$, let us define $Y_{j,p,k}$ to be the set of coordinates of $z$ affecting the function $F_{j,p, k}$. Since each coordinate of $z$ corresponds to $t$ different bit-levels, we let $Y'_{j,p,k} = Y_{j,p,k} \times [t]$ denote the set of entries in the bit-vector $z$ which affect function $F_{j,p, k}$.

We can apply a Discrete Fourier transform (computed via the Fast Walsh-Hadamard Transform algorithm) to each function $F$ with the fixed value $f^0 = y^{\ell}$, to write it as:
$$
F_{j,p,k}(y^{\ell}, z) = \sum_{e \subseteq Y'_{j,p,k}} \gamma_{j,p,k}(e) (-1)^{z \bullet e}
$$
for real parameters $\gamma_{j,p,k}$. Thus,
\begin{align*}
&G(z) = \sum_{j=1}^m \sum_{k_1 = 1}^{m_{j,1}} \sum_{k_2 = 1}^{m_{j,2}} \sum_{\substack{e_1 \subseteq Y'_{j, 1, k_1} \\  e_2 \subseteq Y'_{j,2,k_2}}} \gamma_{j, 1, k_1}(e_1) \gamma_{j,2,k_2}(e_2) (-1)^{z \bullet (e_1 \oplus e_2)}
\end{align*}

For each $j = 1, \dots, m$ construct the ensemble $\mathcal E_j$ defined by 
$$
\mathcal E_j^{p} (e) = \sum_{k: e \subseteq Y'_{j,p, k}} \gamma_{j,p,k}(e)
$$
Thus, the conditional expectation $G(z)$ is precisely $\sum_j S_{\mathcal E_j} (X)$.  We will apply Lemma~\ref{bilinear-ce-prop} to these ensembles.

To do so, we must construct a code $A$ fooling them. All the sets $e \in \langle \mathcal E_j \rangle$ can be viewed as subsets of $[n] \times [t]$ (corresponding to the $n$ variables and $t$ bit-levels). With this interpretation, the subsets of $Y'_{j,p, k}$ can be written as $( \{i_1\} \times t_1) \cup \dots \cup ( \{i_w \} \times t_w)$, where $Y_{j,p,k} = \{i_1, \dots, i_w \}$.  We form the code $A$ via a well-known construction based on van der Monde matrices. By rescaling we may assume wlg that $n = 2^r$. For $\alpha$ a primitive element of the finite field $GF(2^r)$, we define $A$ by
$$
A(i,j) = ( \alpha^j, \alpha^{j+i}, \alpha^{j+ 2 i}, \dots, \alpha^{j + 2 w i} )
$$
Such a code has length $L = O(w \log n)$. Since the rows of a van der Monde matrix are linearly independent over the base field, the vector $A( ( \{i_1\} \times t_1) \cup \dots \cup ( \{i_w \} \times t_w))$ takes distinct values for every choice of $i_1, \dots, i_w, t_1, \dots, t_w$. So $A$ fools all the ensembles. 

All the computational steps up to this point have quasi-complexity $(C_1 + \log(Wn),C_2 + W + n)$. This code $A$ has length $L = O(\log n)$, and so Lemma~\ref{bilinear-ce-prop} runs in time $\tilde O(\log n)$ and using $n^{o(1)} \sum_{j} (m_{j,1} 2^{w t} + m_{j,2} 2^{w t}) = n^{o(1)} W$ processors. It generates a fixed value $z $ with  
$$
G(z) = \bE[S(X) \mid X \langle 1:(\ell+1)t \rangle = y^{\ell} \# z] \geq \bE[ S(X) \mid X \langle 1:\ell t \rangle = y^\ell]
$$
and we set $y^{\ell+1} = y^{\ell} \# z$. At the end, we set $x = y^{\lceil b/t \rceil}$ such that $S(x) \geq \bE[S(X)]$.
\end{proof}

One important application of Theorem~\ref{thm-bilinear-gen} comes from the setting of independent Bernoulli variables.
\begin{definition}
  For any vector of probabilities $q \in [0,1]^n$, we write $X \sim q$ to mean that $X_i$ is distributed as Bernoulli-$q_i$, and all entries of $X$ are independent.
\end{definition}

\begin{theorem}
\label{thm-bilinear-gen2}
Suppose that there are $w$-juntas $F_{j,\ell,k}$, such that for any vector of probabilities $q \in [0,1]^n$ we have
$$
\bE_{X \sim q}[S(X)] = \sum_{j=1}^m \sum_{k_1 = 1}^{m_{j,1}} \sum_{k_2 = 1}^{m_{j,2}} F_{j,1,k_1}(q) F_{j,2,k_2}(q)
$$
and computing all of the functions $F$ has complexity $(C_1, C_2)$. Let $W = \sum_{j=1}^m (m_{j_1} + m_{j_2})$.

Then, for any vector $p \in [0,1]^n$, whose entries are rational number with denominator $2^b$ for $b = O(\log n)$,  there is an algorithm to find $x \in \{0, 1 \}^n$ with $S(x) \geq \bE_{X \sim p} [S(X)]$, with quasi-complexity $(C_1 + \log (Wn), C_2 + W + n)$.
\end{theorem}
\begin{proof}
Define new variables $Y \in \mathcal B_b^n$, and define a function $g: \mathcal B_b^n \rightarrow \{0, 1 \}^n$, in which the $i^{\text{th}}$ coordinate of $g(y)$ is $[p_i 2^b \leq 1]$. If $Y \sim \mathcal B_b^n$, then $g(Y) \sim p$. So $\bE_{X \sim p}[S(X)] = \bE_{Y \sim \mathcal B_b^n} S(g(Y))$, and it suffices to give a bilinear expectations factorization for $S(g(y))$ as a function of $y$.

Suppose we have fixed the most-significant $b'$ bit-levels of $Y$ to some value $y^0$, and the next $t$ bit-levels of $Y$ are set to some varying value $y^1$, and the low-order $b - b' - t$ bit-levels vary freely. Then $g(Y) \sim q$, where the entries $q_i$ are (easily computed) functions of $y^0_i, y^1_i, p_i$. Focusing only on the dependence on $y^1$, we have:
\begin{align*}
\bE[S(g(Y)) \mid Y \langle 1:b' + t \rangle = y^0 \# y^1] &= \bE_{X \sim q(y^1)} [S(X)] = \sum_{j=1}^m \sum_{k_1 = 1}^{m_{j,1}} \sum_{k_2 = 1}^{m_{j,2}} F_{j,1,k_1}(q(y^1)) F_{j,2,k_2}(q(y^1))
\end{align*}

Each function $F$ depends on $w$ coordinates of $q(y^1)$, and each coordinate of $q(y^1)$ depends on only one coordinate of $y^1$.  So each function $F$ depends on at most $w$ bits. This factorization satisfies Theorem~\ref{thm-bilinear-gen}. Hence with quasi-complexity $(\log Wn, W + n)$ we find $y$ such that $S(g(y)) \geq T$; the solution vector is $x = g(y)$.
\end{proof}

\section{Maximal independent set}
\label{mis-sec}
Let $G = (V,E)$ be a graph with $|V| = n$ and $|E| = m$. Finding a maximal independent set (MIS) of $G$ is a fundamental algorithmic building block. The best randomized parallel algorithms are due to Luby in \cite{luby-mis}, which uses $O(m+n)$ processors and $O(\log^2 n)$ time for the EREW PRAM (or $O(\log n)$ time for the CRCW PRAM). The deterministic algorithms have complexity slightly higher than their randomized counterparts. The algorithm of \cite{han} appears to have the lowest processor complexity but requires $O(\log^{2.5} n)$ time on an EREW PRAM, or $O(\log^{1.5} n)$ time on a CRCW PRAM. Other algorithms such as \cite{luby-mis} run in $O(\log^2 n)$ time but require super-linear processor counts.

We will obtain a new algorithm with good processor and time complexity, by showing how an MIS computation can be reduced to maximizing a bilinearizable objective function. Our algorithm is based on Luby's randomized algorithm \cite{luby-mis}. At its heart is a procedure FIND-IS to generate an independent set, which we summarize here.
\begin{algorithm}[H]
\centering
\begin{algorithmic}[1]
\State  Mark each vertex independently with probability $p_v = \frac{c}{d_v}$ for a small constant $c$, and where $d(v)$ denotes the degree of $v$ in the current graph $G$.
\State If any edge $( u, v )$ has both endpoints marked, unmark the endpoint with lower degree. (If both endpoints have the same degree, unmark the one with lower index).
\State  Return the set $I$ of vertices which remain marked.
\end{algorithmic}
\caption{FIND-IS($G$)}
\end{algorithm}

We can find a full MIS by repeatedly calling FIND-IS and forming the residual graph.  For a graph $G$ and an independent set $I$ in $G$, we define $H(I)$ to be the number of edges within distance two of $I$; these are deleted in the residual graph. Given any input graph $G$ with $m$ edges, FIND-IS($G$) produces a random independent set $I$ such that $\bE[H(I)] \geq \Omega(m)$. We will turn this into an NC algorithm by derandomizing the FIND-IS procedure. Specifically, in quasi-complexity $(\log n, m+n)$, we will find an independent set $I$ with $H(I) \geq \Omega(m)$.

To simplify the notation, let us suppose that the vertices have been sorted by degree (breaking ties arbitrarily), so that $d(u) \leq d(v)$ for vertices $u < v$. We let $X(v)$ be the indicator variable that vertex $v$ is marked in line (1). (It may become unmarked at line (2)).

As shown by \cite{luby-mis}, if the vector $X$ is drawn according to an appropriate 2-wise-independent distribution, then $\bE[H(I)] \geq \Omega(m)$. We will explicitly construct a pessimistic estimator of $\bE[H(I)]$ in this case, and we will then show this estimator has a good bilinear factorization. The main idea is to use inclusion-exclusion to estimate the probability that a neighbor of a vertex $v$ is selected for the independent set $I$. We will be careful to ensure that our pessimistic estimator has a simple symmetric form, which is not the case for the estimator used in \cite{luby-mis}.

As is usual in inclusion-exclusion arguments, if the expected number of successes becomes too large, then the success probabilities must be attenuated. Accordingly, for each vertex $v$ we define $G(v) = \sum_{w \in N(v)} 1/d(w)$ and we define the attenuation factor by
$$
a_v = \min(1, 1/G(v))
$$

We now define the random variable $S(v,X)$, which is a pessimistic estimator for the event that $v$ is adjacent to some vertex $w$ selected for the independent set:
$$
S(v, X) = a_v \sum_{w \in N(v)} X(w)  - a_v^2 \sum_{\substack{w<w' \\ w, w'\in N(v)}} X(w) X(w') - a_v \sum_{\substack{w \in N(v) \\ ( w,u) \in E  \\ w < u}} X(w) X(u)
$$
and an overall objective function $S(X) = \sum_v d(v) S(v,X)$.
\begin{proposition}
For any marking vector $X$, the independent set $I_X$ returned by FIND-IS satisfies $H(I_X) \geq S(X)/2$.
\end{proposition}
\begin{proof}
We omit the subscript $X$ for simplicity. If a neighbor of $v$ is placed into $I$, then all of the edges incident to $v$ will be counted in $H(I)$. Each edge can only be counted twice, according to its two endpoints, so 
$$
H(I) \geq 1/2 \times \sum_{v} d(v) \Bigl[ \bigvee_{w \in N(v)} w \in I \Bigr]
$$

Next, we apply an attenuated inclusion-exclusion (as $a_v \leq 1$):
\begin{align*}
\Bigl[ \bigvee_{w \in N(v)} w \in I \Bigr] &\geq a_v \sum_{w \in N(v)} [w \in I] - a_v^2 \sum_{w<w' \in N(v)} [w \in I] [w' \in I] \\
&\geq  a_v \sum_{w \in N(v)} [w \in I] - a_v^2 \sum_{w<w' \in N(v)} X_w X_{w'} 
\end{align*}

A vertex $w$ is placed into $I$ iff it is marked and there is no neighbor $u$ of $w$ such that $u > w$ and $u$ is also marked. So $$
[w \in I] \geq X_w - \sum_{( w,u) \in E, w < u} X_w X_u
$$
and so $[\bigvee_{w \in N(v)} w \in I] \geq S(v,X)$ and the proposition follows.
\end{proof}

\begin{proposition}
\label{exp-prop}
If $X \sim p$ for the vector $p_v = \frac{c}{d(v)}$, then $\bE[S(X)] \geq \Omega(m)$ for $c$ a sufficiently small constant.
\end{proposition}
\begin{proof}
  We omit this proof, since it is very similar to an argument of \cite{luby-mis}. 
\end{proof}

\begin{theorem}
\label{Bfmis-thm1}
There is an algorithm to determine a marking vector $X$ such that $S(X) \geq \Omega(m)$, using quasi-complexity $(\log n, m+n)$.
\end{theorem}
\begin{proof}
We will apply Theorem~\ref{thm-bilinear-gen2} to objective function $S$ to get $S(X) \geq \bE_{X \sim p}[S(X)]$. For any probability vector $q \in [0,1]^n$ we may compute:
{\allowdisplaybreaks
\begin{align*}
\bE_{X \sim q} [S(X)] &= \sum_v d(v) \Bigl( a_v \sum_{w \in N(v)} q_w - \sum  a_v^2 \sum_{\substack{w, w' \in N(v) \\ w < w'}} q_w q_{w'} - a_v \sum_{\substack{w \in N(v) \\ ( w,u) \in E, w < u}} q_w q_u \Bigr) \\
&= \sum_v d(v) \Bigl( a_v \sum_{w \in N(v)} q_w - \sum  a_v^2/2 \sum_{w, w' \in N(v)} q_w q_{w'}  + a_v^2/2 \sum_{w \in N(v)} q_w^2 - a_v \sum_{\substack{w \in N(v) \\ ( w,u) \in E \\ w < u}} q_w q_u \Bigr) \\
&= \sum_v \gamma_v q_v + \sum_v \gamma'_v q_v^2 + \sum_{( u, v ) \in E} \gamma_{uv} q_u q_v - \sum_v d(v)  a_v^2/2 \sum_{w, w' \in N(v)} q_w q_{w'}
\end{align*}
where the values of $\gamma$ are given by
\begin{eqnarray*}
&  \gamma_v = \sum_{w \in N(v)} d(v) a_v, \qquad \gamma'_v = -\sum_{w \in N(v)} d(v) a_v^2/2 \qquad \text{for $v \in V$} \\
& \gamma_{uv} = -\sum_{w \in N(u)} d(v) a_v \qquad \qquad \text{for $u < v, ( u, v ) \in E$}
\end{eqnarray*}
}

Let us compute the weight $W$ of this bilinear-expectations factorization. First, we have terms for each vertex and for each edge; they contribute $O(m + n)$. Next, for each vertex $v$, we have a bilinear form where the two relevant sets $w, w'$ vary over the neighbors of $v$; this contributes $2 d(v)$. Overall $W = O( m + n + \sum_v d(v) ) = O(m)$. So Theorem~\ref{thm-bilinear-gen} allows us to determine $X$ with $S(X) \geq \Omega(m)$ with quasi-complexity $(\log n, m+n)$.
\end{proof}

\begin{theorem}
\label{Bfmis-thm2}
There is an algorithm to compute the MIS of graph $G$ using quasi-complexity $(\log^2(mn), m+n)$.
\end{theorem}
\begin{proof} 
First, using a simple pre-processing step with complexity $(\log mn, m + n)$, we may remove any duplicate edges. Thus we assume $m \leq n^2$. Next repeatedly apply Theorem~\ref{Bfmis-thm1} to find a marking vector $X$ with $S(X) \geq \Omega(m)$. The resulting independent set $I$ satisfies $H(I) \geq \Omega(m)$. After $O(\log m)$ iterations of this procedure, the residual graph has no more edges, and we have found an MIS.
\end{proof}

Luby's randomized MIS algorithm has complexity $(\log^2(mn), m+n)$ on an EREW PRAM, so this is a nearly optimal derandomization.

\section{The fourth-moment method}
\label{Bffourth-moment-sec}
Berger \cite{berger} introduced the fourth-moment method as a general tool for showing lower bounds on $\bE[ |X| ]$ for certain types of random variables,  based on comparing the relative sizes of $\bE[X^2]$ and $\bE[X^4]$. This can lead to NC algorithms for selecting large $|X|$ which are amenable to bilinear factorization.

\subsection{Gale-Berlekamp Switching Game}
\label{Bfgale-berlekamp-sec}
We begin with a very straightforward application, the \emph{Gale-Berlekamp Switching Game}. In this case, we are given an $n \times n$ matrix $a$, whose entries are $\pm 1$. We would like to find vectors $x, y \in \{-1, +1 \}^n$, such that $\sum_{i,j} a_{i,j} x_i y_j \geq \Omega(n^{3/2})$.  Brown \cite{brown} and Spencer \cite{spencer} gave deterministic sequential polynomial-time algorithms to construct such $x, y$, using the method of conditional expectations, which Berger \cite{berger} transformed into an NC algorithm with complexity $(\log n, n^4)$. We will reduce the processor count by a factor of $n$.

\begin{theorem}
\label{gb-game1}
There is an algorithm with quasi-complexity $(\log n, n^3)$ to find $x, y \in \{-1, +1\}^n$ satisfying $\sum_{i,j} a_{i,j} x_i y_j \geq n^{3/2}/\sqrt{3}$.
\end{theorem}

\begin{proof}
Following \cite{berger}, we define $R_i(y) = \sum_j a_{ij} y_j$ and then set $x_i = (-1)^{[R_i(y) < 0]}$. Thus
$$
\sum_{i,j} a_{i,j} x_i y_j = \sum_i x_i R_i(y) = \sum_i |R_i(y)|
$$
so we have reduced this problem to maximizing $\sum_i |R_i(y)|$.

We now make us of the the following fact, which is the heart of the fourth-moment method: for any integer $Z$ and any real $q > 0$, we have $|Z| \geq \frac{3 \sqrt{3}}{2 \sqrt{q}} (Z^2 - Z^4/q)$. Thus, let us define the objective function
\begin{equation}
\label{Bfs-eqn}
S(y) = \sum_i \frac{3 \sqrt{3}}{2 \sqrt{q}} (R_i(y)^2 - R_i(y)^4/q) 
\end{equation}
for $q = 3(1+3 n)$. It suffices to construct $S(y) \geq \Omega(n^{3/2})$.

When $y$ is drawn uniformly from $\{-1, +1 \}^n$, then our choice of $q$ ensures:
\begin{align*}
  \bE[S(y)] &= \sum_i \frac{3 \sqrt{3}}{2 \sqrt{q}} (\bE[R_i(y)^2] - \bE[R_i(y)^4]/q) = n^2 (3 n + 1)^{-1/2} \geq \Omega(n^{3/2})
\end{align*}

To show that the function $S(y)$ is bilinearizable, we factor it a :
\begin{align*}
S(y) &= \sum_i \frac{3 \sqrt{3}}{2 \sqrt{q}} \Bigl( \sum_{j_1, j_2} a_{i,j_1} a_{i,j_2} y_{j_1} y_{j_2} (1 - q^{-1} \sum_{j_3, j_4} a_{i,j_3} a_{i, j_4} y_{j_3} y_{j_4} ) \Bigr)
\end{align*}

This has the form required by Definition~\ref{def-bilin}. In particular, for each value $i = 1, \dots, n$ we get functions $F^1_{k_1}, F^2_{k_2}$ where $k_1$ enumerates over pairs $j_1, j_2$ and $k_2$ enumerates over pairs $j_3, j_4$. In this case, $b = t = 1$ and the overall weight is $W = O(n^3)$. So Theorem~\ref{thm-bilinear-gen} runs in quasi-complexity $(\log n, n^3)$. 
\end{proof}

\subsection{Maximum acyclic subgraph}

\label{maximum-acyclic-sec}
In our next example, we will have to work much harder to transform our objective function into the desired bilinear form.

Given a directed graph $G = (V,A)$, we consider the \emph{maximum acyclic subgraph} problem, which is to find a maximum-sized subset of arcs $A' \subseteq A$, such that $G' = (V, A')$ is acyclic. In \cite{berger}, Berger gives a procedure for finding a relatively large (though not largest) acyclic set $A'$, satisfying
\begin{equation}
\label{max-acyclic-eq0}
|A'| \geq |A|/2 + \Omega \Bigl( \sum_{v \in V} \sqrt{d(v)} + |d_{\text{out}}(v) - d_{\text{in}}(v)| \Bigr)
\end{equation}

Here is a sketch of the algorithm. First, assign every vertex a random rank $\rho(v)$ in the range $\{1, \dots, \sigma \}$, for some $\sigma = n^{O(1)}$. Then, for each vertex $v$, put into $A'$ either the set of out-edges from $v$ to higher-ranking $w$, \emph{or} the set of in-edges to $v$ from higher-ranking $w$, whichever is larger. The resulting edge set $A'$ is acyclic. As shown in \cite{berger} and \cite{crs}, this yields
\begin{equation}
\label{max-acyclic-eq1}
|A'| \geq |A|/2 - \sum_{( u, v ) \in A} [\rho(u) = \rho(v)] + \sum_{v,  S \subseteq N(V), |S| \leq 4} c_{v,S} [\rho (v) < \min_{w \in S} \rho(w)]
\end{equation}

The RHS of (\ref{max-acyclic-eq1}) can be viewed as an objective function $S(\rho)$ in terms of the ranks $\rho$. It is a sum over all vertices $v$ and all sets $S$ of its neighbors (either in-neighbors or out-neighbors) of cardinality at most $4$. Furthermore, if $\rho$ is drawn from a 5-wise independent probability distribution, then the RHS has good expectation:
$$
\bE[S(\rho)] \geq |A|/2 + \Omega \Bigl(\sum_{v \in V} \sqrt{d(v)} + |d_{\text{out}}(v) - d_{\text{in}}(v)| \Bigr)
$$

By using conditional expectations on this quantity, \cite{berger} achieves an NC algorithm with complexity $(\log^3 n, n \Delta^4)$. This was improved by \cite{crs} to a complexity of $(\log n, n \Delta^4)$. The processor cost in both cases comes from evaluating all $n \Delta^4$ summands of the RHS of (\ref{max-acyclic-eq1}), for any putative $\rho$.
\begin{theorem}
There is an algorithm with quasi-complexity $(\log mn, n \Delta^2)$ to find an arc set $A' \subseteq A$ satisfying (\ref{max-acyclic-eq0}).
\end{theorem}
\begin{proof}
  We will show how to get a bilinear factorization for $S(\rho)$. The term $\sum_{( u, v ) \in A} [\rho(u) = \rho(v)]$ is easy to write in this form, as it has only $O(n \Delta)$ summands altogether and they can listed individually. The challenge is dealing with $\sum_{v,  S \subseteq N(V), |S| \leq 4} c_{v,S} [\rho (v) < \min_{w \in S} \rho(w)]$.  Say we have fixed a vertex $v$; to simplify the discussion, suppose $N_{\text{in}}(v) \cap N_{\text{out}}(v) = \emptyset$.

In \cite{berger}, the sum over $S$ is decomposed into a constant number of cases, depending on how many in-neighbors and how many out-neighbors of $v$ are selected; the coefficient $c_{v,S}$ depends on the number of in-neighbors and out-neighbors in $S$, not their precise identities. Thus, we can write for instance
\begin{equation}
\label{max-acyclic-sum1}
S(\rho) =  \sum_v c_1 \sum_{\substack{w_1 < w_2 \in N_{\text{in}}(v) \\ w_3 < w_4 \in N_{\text{out}}(v)}} [\rho_v < \rho_{w_1}] [\rho_v < \rho_{w_2}] [\rho_v < \rho_{w_3}] [\rho_v < \rho_{w_4}] + \dots
\end{equation}
(Here, the summand we are displaying is representative of the remaining, elided summands.)

We wish to show that this $S(\rho)$ has bilinear expectations with window $t = \frac{\log n}{\log \log n}$. So, suppose that $\rho \langle 1:b'+t \rangle$ is fixed to some value $\rho^0 \# \rho^1$, and all other entries of $\rho$ are independent Bernoulli-$1/2$. The comparison $\rho(v) < \rho(w)$ is only undetermined for $\rho^0(v) = \rho^0(w)$, and so we have
\begin{equation}
\label{max-acyclic-sum2}
\begin{aligned}
S(\rho) =  \sum_v c_1 \sum_{\substack{w_1 < w_2 \in N_\text{in}(v) \\ w_3 < w_4 \in N_{\text{out}}(v) \\ \rho^0(w_1) = \rho^0(w_2) = \rho^0(w_3) = \rho^0(w_4)}} [\rho^1_v < \rho^1_{w_1}] [\rho^1_v < \rho^1_{w_2}] [\rho^1_v < \rho^1_{w_3}] [\rho^1_v < \rho^1_{w_4}]   \\
\qquad + c_2 \sum_{\substack{w_1 \in N_\text{in}(v) \\ w_3 < w_4 \in N_{\text{out}}(v)\\ \rho^0(w_1) = \rho^0(w_3) = \rho^0(w_4)}} \sum_{\substack{w_2 \in N_{\text{in}}(v) \\ \rho^0(w_1) > \rho^0(v)}} [\rho^1_v < \rho^1_{w_2}] [\rho^1_v < \rho^1_{w_3}] [\rho^1_v < \rho^1_{w_4}] + \dots
\end{aligned}
\end{equation}

(This multiplies the number of cases for the summation by a constant factor. All of the summands have a similar form, so we will write only one, representative summand, henceforth.)

Next, sum over all $2^t$ possible values for $\rho^1(v)$ to obtain:
\begin{equation}
\label{max-acyclic-sum3}
S(\rho) =  \sum_v \sum_{z=0}^{2^t - 1} c_1 [\rho^1_v = z] \negthickspace \negthickspace \sum_{\substack{w_1 < w_2 \in N_\text{in}(v) \\ w_3 < w_4 \in N_{\text{out}}(v) \\ \rho^0(w_1) = \rho^0(w_2) = \rho^0(w_3) = \rho^0(w_4)}} \negthickspace \negthickspace [z < \rho^1_{w_1}] [z < \rho^1_{w_2}] [z < \rho^1_{w_3}] [z < \rho^1_{w_4}] + \dots
\end{equation}

Noting that $w_1, w_2, w_3, w_4$ are distinct, we compute the conditional expectation of (\ref{max-acyclic-sum3}):
\begin{equation*}
\label{max-acyclic-sum4}
\bE[S(\rho)] =  \sum_v \sum_{z=0}^{2^t - 1} c_1 P(\rho^1_v = z) \sum_{\substack{w_1 < w_2 \in N_\text{in}(v) \\ w_3 < w_4 \in N_{\text{out}}(v) \\ \rho^0(w_1) = \rho^0(w_2) = \rho^0(w_3) = \rho^0(w_4)}} P(z < \rho^1_{w_1}) P(z < \rho^1_{w_2}) P(z < \rho^1_{w_3}) P(z < \rho^1_{w_4}) + \dots
\end{equation*}

Finally, remove the condition that $w_1 < w_2, w_3< w_4$ via inclusion-exclusion:
\begin{equation*}
\begin{aligned}
\bE[S(\rho)] =  \sum_v \sum_{z=0}^{2^t - 1} \frac{c_1}{4} P(\rho^1_v = z) \sum_{\substack{w_1, w_2 \in N_\text{in}(v) \\ w_3, w_4 \in N_{\text{out}}(v) \\ \rho^0(w_1) = \rho^0(w_2) = \rho^0(w_3) = \rho^0(w_4)}} P(z < \rho^1_{w_1}) P(z < \rho^1_{w_2}) P(z < \rho^1_{w_3}) P(z < \rho^1_{w_4}) \\
 \qquad - \sum_v \sum_{z=0}^{2^t - 1} \frac{c_1}{4} P(\rho^1_v = z) \sum_{\substack{w_{12} \in N_\text{in}(v) \\ w_3, w_4 \in N_{\text{out}}(v) \\ \rho^0(w_{12}) = \rho^0(w_3) = \rho^0(w_4)}} P(z < \rho^1_{w_{12}})^2 P(z < \rho^1_{w_3}) P(z < \rho^1_{w_4}) + \dots
\end{aligned}
\end{equation*}

The number of types of different summands has now multiplied to a painfully large constant. However, they are all essentially equivalent to each other: they all have the form 
$$
P(\rho^1_v = z) \sum_{\substack{w_1, w_2 \\ w_3, w_4}} P(z < \rho^1_{w_1}) P(z < \rho^1_{w_2}) P(z < \rho^1_{w_3})  P(z < \rho^1_{w_4})
$$
where $v$ ranges over the vertices of $G$, $z$ ranges over $0, \dots, 2^{t}-1$, and the vertices $w_1, w_2, w_3, w_4$ come from subsets of neighbors of $v$. In particular, this has the bilinear form required by Theorem~\ref{thm-bilinear-gen}. The weight of this bilinear factorization is given by $W = O( c \times n \times 2^t \times \Delta^2 )$,
where  $c$ is the number of different cases in the summation (a constant), $n$ reflects that $v$ varies, $2^t$ reflects that $z$ varies, and $\Delta^2$ reflects the number of possible choices for $w_1, w_2$ or $w_3, w_4$. We apply Theorem~\ref{thm-bilinear-gen} to the objective function $S$ and variables $\rho$ using quasi-complexity $(\log (mn), W) = (\log mn, n \Delta^2)$.
\end{proof}

\section{Discrepancy}
\label{symmetric-sec}
Consider a system of $n$ variables $x_1, \dots, x_n \in \mathcal B_b$ and $m$ linear functionals of the form
$$
l_j(x) = \sum_i l_{j,i}( x_i )
$$
and let $\mu_j$ be the expectation of $\bE[l_j(X)]$ for $X \sim \mathcal B_b^n$.  The problem of \emph{discrepancy minimization} is to find a value $x \in \mathcal B_b^n$ such that all the values $l_1(x), \dots, l_m(x)$ are  simultaneously close to their means $\mu_1, \dots, \mu_m$. (The precise definition of closeness may depend on the application.)

If components of $X$ are chosen \emph{independently}, then indeed $l_j(X)$ will be concentrated tightly around $\mu_j$, with exponentially small probabilities of deviation. When $X$ is drawn from a $w$-wise-independent probability space, then weaker bounds are available. One particularly important type of tail bound, due to Bellare-Rompel \cite{bellare-rompel}, is derived by noting that
$$
|l_j(X) - \mu_j| > a_j \Leftrightarrow (l_j(X) - \mu_j)^r > a_j^r
$$
and furthermore, by Markov's inequality
$$
P((l_j(X) - \mu_j)^r > a_j^r) \leq \bE[ (l_j(X) - \mu_j)^r ] / a_j^r
$$

For example, when $r = 2$, this is precisely Chebyshev's inequality. For larger values of $w$, one can obtain a concentration which is similar to the Chernoff bound; the main difference is that instead of an exponential dependence on the deviation, one has only a dependence of order $\approx a^{1/w}$.

A common proof technique to satisfy multiple discrepancy constraints, is to apply a concentration inequality to upper-bound the probability of violating each constraint separately and then a union bound to sum these probabilities over all constraints.  If the total probability of violating the constraints is less than one, then with positive probability, they can simultaneously be satisfied. Crucially, such concentration inequalities use \emph{symmetric} polynomials applied to the $l_j(X)$ variables.

\begin{definition}[Symmetric-moment bounds]
  Suppose that symmetric polynomials $Q_1(\vec z) \dots, Q_m(\vec z)$ in $n$ variables of degree at most $d$ satisfy the condition
  $$
  \sum_j \bE[ Q_j( l_{j1}(X(1)), \dots, l_{jn}(X(n)) ] < 1
  $$

  Then we say that the condition $\sum_j  Q_j( l_{j1}(x(1)), \dots, l_{jn}(x(n))) < 1$ defines a \emph{degree-$d$ symmetric-moment bound} for the linear function $l$.
\end{definition}

For example, the Bellare-Rompel type tail bounds would use $Q_j(z_1, \dots, z_n) = \frac{ (\sum_i z_i - \mu_j)^d }{  a_j^d }$. In order to find a vector $x$ satisfying the discrepancy condition, it thus suffices to find $x \in \mathcal B_b^n$ with
\begin{equation}
\label{cond1}
S(x) = \sum_j Q_j( l_{j1} (x(1)), \dots, l_{jn} (x(n)) ) < 1
\end{equation}
and further we know that $\bE[S(X)] < 1$ when $X$ is chosen from the appropriate distribution. In other words, we have reduced the problem of satisfying discrepancy constraints to satisfying an appropriate symmetric-moment bound.

Berger \& Rompel \cite{berger-simulating} and Motwani et al. \cite{motwani-naor-naor} discussed the use of conditional expectations to simulate the concentration bounds corresponding to $w$-wise-independence in this context. The processor complexity of their procedure is roughly $O(m n^w)$. We next show how to reduce the processor complexity to roughly $O(m n^{\lceil w/2 \rceil})$.

\begin{theorem}
\label{discrepancy-thm}
Suppose we have symmetric-moment bounds of degree $d = O(1)$, and where $b = O(\log n)$. Furthermore suppose that we have a ``partial-expectations oracle (PEO),'' namely that with complexity $(C_1, C_2)$  we can compute for all $j, i$ and $s \leq d$, and any vector $z \in \{0, 1 \}^{b'}$ for $b' \leq b$,  the quantities 
$$
\bE_{X \sim \mathcal B_b} [l_{j,i} (X)^s \mid X \langle 1: b' \rangle = z]
$$

Then there is an algorithm with quasi-complexity $(C_1 + \log(mn), C_2 + m n^{\lceil d/2 \rceil})$ to find a vector $x \in \mathcal B_b^n$ satisfying the discrepancy constraints.
\end{theorem}
\begin{proof}  
  Since the summands $Q_j$ are symmetric polynomials, then for every monomial in the sum $S(x)$, all corresponding monomial with permuted indices  must appear as well. So the sum $S(x)$ from (\ref{cond1}) can be written in the form
  \begin{equation}
\label{cond2}
S(x)  = \sum_{j=1}^m \sum_{v}  \gamma_{j,v} \sum_{\substack{i_1,  \dots , i_d\\\text{distinct}}}l_{j,i_1} (x_{i_1})^{s_{v,1}} \dots l_{j,i_d}( x_{i_d} )^{s_{v,d}}
  \end{equation}
 Here, $v$ ranges over the types of monomials, with associated exponents $s_{v,1}, \dots, s_{v,d}$ which are non-negative integers with $s_{v,1} + \dots + s_{v,d}  \leq d$. For constant $d$, we have $v = O(1)$.

Let us flatten the indexing in the sum (\ref{cond2}); instead of double indexing by $j$ and $v$, we index by a new variable $k$. Since $v$ has $O(1)$ possible values, the index $k$ can take on at most $m' = O(m)$ values, and so may write our objective function more compactly as 
\begin{equation}
\label{cond3}
S(x) = \sum_{k=1}^{m'}  \gamma_{k} \sum_{\substack{i_1,  \dots , i_d\\\text{distinct}}} l_{k,i_1} (x_{i_1})^{s_{k,1}} \dots l_{k,i_d}( x_{i_d} )^{s_{k,d}}
\end{equation}

 We will use Theorem~\ref{thm-bilinear-gen} to find $x \in \mathcal B_b^n$ with $S(x) < 1$. This requires a bilinear-expectations factorization of $S$. So suppose we have fixed $X \langle 1, \dots, b' + t \rangle$, and we wish to compute the expectation of $S(X)$ when the least-significant bits of $X$ vary. We have
\begin{equation}
\label{cond4}
\bE[S(X)] = \sum_{k=1}^{m'} \gamma_j \sum_{\substack{i_1,  \dots , i_d\\\text{distinct}}}  \bE[l_{k,i_1} (X_{i_1})^{s_{k,1}}] \dots \bE[l_{k,i_d}( X_{i_d} )^{s_{k,d}} ]
\end{equation}
(we are omitting, for notational simplicity, the conditioning of the most-significant bits of $X$.)

We may remove the restriction that $i_1, \dots, i_d$ are distinct via a series of inclusion-exclusion expansions. For example, we subtract off the contribution coming from $i_1 = i_2$, the contribution from $i_1 = i_3$ etc., and then add in the contribution from $i_1 = i_2 = i_3$ etc. This transformation blows up the number of summands, again by a factor which is constant for $d = O(1)$.   This gives us an objective function in the form
\begin{equation}
\label{cond5}
\bE[S(X)] = \sum_{k=1}^{m''} \gamma''_k \sum_{i_1,  \dots , i_d}  \bE[l_{k,i_1} (X_{i_1})^{s_{j,1}}]^{r_{j,1}} \dots \bE[l_{k,i_d}( X_{i_d} )^{s_{k,d}}]^{r_{k,d}}
\end{equation}
(where the weights $\gamma''$ and number of summands $m''$ have changed again, and $m'' = O(m)$). Each term in (\ref{cond5}) involves two different exponents, one inside the expectation and one outside.

The PEO allows us to compute each term $\bE[X_i^s]^r$. We then can put (\ref{cond5}) into the form of Theorem~\ref{thm-bilinear-gen} by splitting each summand into two groups, one in which  $i_1, \dots, i_{\lfloor d/2 \rfloor}$ vary and the other in which $i_{\lfloor d/2 \rfloor + 1}, \dots, i_d$ vary.  Thus, this factorization has  weight $W = O(m n^{\lceil d/2 \rceil})$. 
\end{proof}

\section{Fooling automata}
\label{Bfsec:fool}
On their own, conditional expectations or low-independence probability spaces do not lead to derandomizations of many algorithms. For example, we have seen in Section~\ref{symmetric-sec} that low-independence probability spaces achieve much weaker discrepancy bounds than fully-independent spaces.  Sivakumar \cite{sivakumar} described an alternate derandomization method based on \emph{log-space statistical tests}. To summarize, consider a collection of statistical tests which have good behavior on random inputs (e.g. these statistical tests are the discrepancy of each linear functional), and which can be computed in $O(\log n)$ memory. Then there is an NC algorithm to build a relatively small probability distribution which has similar behavior with respect to the these statistical tests. This probability distribution can then be searched exhaustively.

The Sivakumar derandomization method is based on an algorithm of Nisan \cite{nis1,nis2} for constructing probability spaces which fool a given list of finite-state automata. Mahajan et al. \cite{mrs} further optimized this for a special class of automata based on counters. Our presentation follows \cite{mrs} with a number of technical modifications.

\textbf{Basic definitions.} We consider $m$ finite-state automata, which are driven by $T$ random variables $r_1, \dots, r_T$ which we refer to as the \emph{driving bits}; we will suppose that these random variables $r_1, \dots, r_T$ are fair coin flips, and that $T$ is a power of two. (This setup does not lose much generality). Each automaton has a state space $A$, of size $|A| = \eta$; this includes a designated start state denoted state $0$. The automata also have a ``clock'' counter, which simply advances by one every time-step, and is not included in the state space $S$. We define the transition function $F$ as follows: suppose automaton $\# i$, currently in step $s$ and time $t$, receives bit $r = r_t$; then it transitions to state $F(i,r,s,t)$ and time $t+1$. We assume  throughout that the transition function $F$ can be computed, say using $O(1)$ time and processors. In order to simplify the logarithmic notations, we define $n = \max(T, \eta, m, 1/\epsilon)$.

Handling multiple automata ($m > 1$) is important for our applications, and it is crucial here that all automata under consideration process the driving bits \emph{in the same fixed order}.

We abuse notation, so that $F$ can refer also to multi-step transitions. For any $\ell \leq T$ and bitstream $r$ of length $\ell$, we define $F^{\ell} (i,r,s,t)$ to be the result of transiting on automata from times $t$ to $t+\ell$. Thus, we can write the entire trajectory of the automata as $F^T (0, r, 0, 0)$.

Our goal is to find a distribution $\tilde D$ on the driving bits $r_1, \dots, r_T$ to ``fool'' these automata. That is, the behavior of the automata when presented with $r \sim \tilde D$ should be close to the behavior when $r \sim \{0, 1 \}^T$. For the purposes of our algorithm, we will measure error as 
$$
\text{Err}_{t,t+h} (D, \tilde D) = \max_{\substack{i \in [m] \\ s \in S}} \sum_{s' \in S} \Bigl| P_{r \sim D}( F^h(i,r, s, t) = s')- P_{r \sim \tilde D} (F^h(i, r, s, t) = s') \Bigr|
$$

\textbf{Stepping tables.} In our algorithms, we will manipulate distributions over driving bits $r$ over intermediate time intervals $t$ to $t+h$. It is useful to maintain a data structure we refer to as the \emph{stepping table} $S_{t,h}(r)$; this lists $F^h(i, r, s, t)$ for every value $s \in A$. Given $S_{t,h}(r_1)$ and $S_{t+h,h}(r_2)$, we can form $S_{t,2h}(r_1 \# r_2)$ with quasi-complexity $(\log n, m \eta)$.  Whenever we construct a distribution $D$ over driving bits from time $t$ to $t+h$, we will also maintain the value of $S_{t,h}(r)$ for every $r \in D$, which we write more compactly as $S_{t,h}(D)$. At several points in the algorithm, we will discuss, for example, how to transform one distribution $D$ into another $D'$; we also assume that $S_{t,h}(D)$ is given as input and $S_{t', h'}(D')$ will be produced as output.

\subsection{The REDUCE algorithm}
The REDUCE subroutine is the technical core of the automata-fooling procedure. It takes two distributions $ D_1,  D_2$, which go from times $t$ to $t+h$ and $t+h$ to $t + 2h$ respectively,  and returns a single distribution $ D$ which is close to the product distribution $ D_1 \times  D_2$. The latter distribution would have support $|D_1| |D_2|$, which is too large; the REDUCE subroutine compresses it into a new distribution $D$ with much smaller support.

\begin{theorem}
\label{thm-reduce}
Let $D_1,  D_2$ be distributions of size $n^{O(1)}$. Then there is an algorithm $\text{REDUCE}(D_1, D_2, \epsilon)$ to construct a distribution $ D$ of size $O(m \eta^2 \epsilon^{-2})$ satisfying $\text{Err}_{t, 2 h}( D,  D_1 \times  D_2) \leq \epsilon$. It has quasi-complexity $(\log n, m (m \eta^4 \epsilon^{-2} + (| D_1| + |  D_2|) (\eta h + \eta^2)))$.
\end{theorem}
\begin{proof}
Let $E = 2 m \eta^2 \epsilon^{-2}$. We will form $D$ by selecting $E$ pairs $( x, y ) \in D_1 \times D_2$; then $D$ is the uniform distribution on the strings $\{ x_1 \# y_1, \dots, x_E \# y_E \}$. We define variables $z_j = ( x_j, y_j )$, where $j = 1, \dots, E$; this designates that the $j^{\text{th}}$ string in $D$ is given by concatenating the $x_j^{\text{th}}$ string in $ D_1$ with the $y_j^{\text{th}}$ string in $ D_2$. Each $z_j$ variable is thus a bitstring of length $b = \log_2 | D_1| + \log_2 | D_2| = O(\log n)$. At the end of this process, the stepping tables $S_{t,h}(D_1)$ and $S_{t+h, h}(D_2)$ can be used to form $S_{t, 2h}(D)$ with quasi-complexity $(\log n, m E \eta)$.

Let us fix an automaton $i$, a start time $t$, and a transition length $h$. For any pair of states $(s,s')$, we define  $G_{s,s'}(r) = [ F^{2h}(i, r, s, t) = s' ]$ and  $p(s,s') = P_{r_1 \sim  D_1, r_2 \sim  D_2}(G_{s,s'}(r_1 \# r_2))$. For a given choice of the distribution $D$, we also define $q(s,s', D) = P_{r \sim D} (G_{s,s'}(r))$. In order to ensure that $D$ has the desired error, we claim that it suffices to show that every pair $s,s'$ has
\begin{equation}
\label{reduce-eqn2}
| q(s,s', D) - p(s,s') | \leq \epsilon \sqrt{p(s,s')}/\sqrt{\eta}
\end{equation}

For then, summing over $s'$, 
{\allowdisplaybreaks
\begin{align*}
\text{Err}_{t, 2 h}( D,  D_1 \times  D_2) &= \max_{\substack{i \in [m] \\ s \in S}} \sum_{s' \in S} \Bigl| P_{r \sim  D} (F^h(i, r, s, t) = s') - P_{r \sim  D_1 \times  D_2} (F^h(i, r, s, t) = s') \Bigr|  \\
&\leq \epsilon/\sqrt{\eta} \sum_{s' \in S}  \sqrt{ P_{r \sim  D_1 \times  D_2} (F^h(i, r, s, t) = s') } \leq \epsilon/\sqrt{\eta} \sum_{s' \in S} \sqrt{1/\eta} \leq \epsilon
\end{align*}
}
Now, suppose that we set the variables $z_j$ equal to independent unbiased random variables $Z_j$. Then $q(s,s', D)$ is (up to scaling) simply the number of indices $j$ satisfying $G_{s,s'}(Z_j)$.  Let $\mathcal E_{s,s'}$ denote the bad event that  $| q(s, s', D) - p(s,s')| > \epsilon \sqrt{p(s,s')}/\sqrt{\eta}$. We wish to avoid the event $\mathcal E_{s,s'}$, for every pair of states $s,s'$ as well as every automaton $i$.

The random variable $q(s, s', D)$ has mean $p(s, s')$ and variance $p(s, s') (1-p(s,s'))/|E|$. So by Chebyshev's inequality
$$
P(\mathcal E_{s,s'}) \leq \frac{p(s,s') \eta (1-p(s,s'))}{E \epsilon^2 p(s,s') } \leq \frac{\eta}{E \epsilon^2}
$$
By taking a union bound over $s, s', i$, we see that there is a positive probability that no $\mathcal E$ is true as long as $\eta m \times \eta E^{-1} \epsilon^{-2} < 1$, which holds by hypothesis.

Thus, using degree-$2$ symmetric-moment bounds (Chebyshev's inequality), we have shown that $D$ has the desired properties with positive probability. Each $s, s', i$ gives rise to a linear functional (\ref{reduce-eqn2}), so we have a total of $\eta^2 m$ functionals. The variable set here is the list of bit-strings in $D$, comprising $E$ total variables of length $b$. Thus Theorem~\ref{discrepancy-thm} can be used to constructively match  Chebyshev's inequality with quasi-complexity $(C_1 + \log n, C_2 + \eta^2 m E)$ where $(C_1, C_2)$ is the cost of the appropriate PEO.

\textbf{The PEO.}
To complete this algorithm, we will build an appropriate PEO for Theorem~\ref{discrepancy-thm} with quasi-complexity $(\log n, m \eta^2 + (| D_1| + | D_2|) m \eta h)$. Here, for any $i, s, s'$ and $z = ( x, y )$, the corresponding coefficient in the linear function $l_{i,s,s'}(z)$ is the indicator $[F^{2 h}(i, x \# y, s, t) = s']$. As these coefficients are always zero or one, we only need to compute expectations of the form $\bE[l_{i,s,s'} (Z_j) \mid Z_j \langle 1 : b' \rangle = u_j]$.

Each variable $Z_j$ here consists of two coordinates $X, Y$; here $X$ denotes the choice of the transition in $ D_1$ and $Y$ denotes the choice of the transition in $ D_2$. Thus, either the most-significant bit-levels of $X$ are fixed, while its least-significant bits and all of $ Y$ are free to vary; or $ X$ is fixed and the most-significant bit-levels of $ Y$ are fixed while its least-significant bits vary.  We discuss the PEO in the case of varying $X$; the case of varying $Y$ is similar and omitted.

Using $S_{t+h, h}(D_2)$, we construct an $\eta \times \eta$ transition probability matrix $M$ from time $t+h$ to $t+2h$, recording the probability of each state transition for $r_2 \sim D_2$. This step can be implemented with quasi-complexity $(\log n, |D_2| m \eta h + m \eta^2)$. Next, loop over $r_1 \in D_1$ and use the stepping table $S_{t,h}(D_1)$ to map each state $s \in A$ to the intermediate state $s'' = F^h(i,r_1, s, t)$. We then use $M$ to determine the probability that $s''$ transitions to any state $s'$ at time $t+2 h$. We sum these counts over the fixed most-significant bit-levels of $X$. Overall this step has quasi-complexity $(\log n, | D_1| (m \eta h + m \eta^2))$.
 \end{proof}
 
\subsection{The FOOL algorithm}
We build the automata-fooling distribution recursively through a subroutine $\text{FOOL}(t, h, \epsilon)$, which generates a distribution $\tilde D_{t,h}$ fooling the automata to error $\epsilon$ for the transitions from time $t$ to $t+h$.

\begin{algorithm}[H]
\centering
\begin{algorithmic}[1]
  \State  If $h = 1$, then return the uniform distribution.
  \State Otherwise, in parallel recursively execute
  $$\tilde D_1 = \text{FOOL}(t,h/2,\epsilon/2 (1 - 1/h)), \qquad
  \tilde D_2 = \text{FOOL}(t+h/2, h/2, \epsilon/2 (1 - 1/h))
  $$
  \State Compute $\tilde D = \text{REDUCE}(\tilde D_1, \tilde D_2, \frac{\epsilon}{h})$
  \State Return $\tilde D$
\end{algorithmic}
\caption{FOOL($t,h,\epsilon$)}
\end{algorithm}

To analyze this process, we simply need to show that the accumulation of errors is controlled. We quote a useful error-accumulation lemma from \cite{mrs}:
\begin{lemma}[\cite{mrs}]
\label{err-1}
Let $D_1, D_1'$ be distributions from times $t$ to $t + h$, and let $D_2, D_2'$ be distributions from times $t + h$ to $t + h + h'$. Then
$$
\text{Err}_{t,h+h'}(D_1 \times D_2, D_1' \times D_2') \leq \text{Err}_{t,h}(D_1, D_1') + \text{Err}_{t+h, h'} (D_2, D_2')
$$
\end{lemma}

\begin{proposition}
Let $\tilde D = \text{FOOL}(t,h, \epsilon)$ and let $U$ be the uniform distribution over $\{0, 1\}^h$. Then $\text{Err}(\tilde D, U) \leq \epsilon$.
\end{proposition}
\begin{proof}
We prove this by induction on $h$. When $h = 1$, this is vacuously true as $\tilde D$ is itself the uniform distribution. For $h > 1$, by inductive hypothesis, $\tilde D_1, \tilde D_2$ both have error $\epsilon/2 (1 - 2/h)$ with respect to their uniform distributions, and $\tilde D$ has error $\frac{\epsilon}{h}$ with respect to $\tilde D_1, \tilde D_2$. So:
\begin{align*}
\text{Err}_{t,h} (\tilde D, U) &\leq \text{Err}_{t,h} (\tilde D, \tilde D_1 \times \tilde D_2) + \text{Err}_{t,h} (\tilde D_1 \times \tilde D_2, U)  \qquad \text{triangle inequality} \\
&\leq \text{Err}_{t,h} (\tilde D, \tilde D_1 \times \tilde D_2) + \text{Err}_{t,h} (\tilde D_1, U) + \text{Err}_{t+h, h} (\tilde D_2, U)  \qquad \text{by Lemma~\ref{err-1}} \\
&\leq \frac{\epsilon}{h} + \epsilon/2 (1 - 1/h) + \epsilon/2 (1 - 1/h) =\epsilon.
\end{align*}
\end{proof}

\begin{theorem}
\label{thm-fool}
The FOOL algorithm has quasi-complexity $(\log T \log n, m^2 \eta^4 T^3 \epsilon^{-2})$. The resulting distribution $\tilde D$ has support $|\tilde D| = O(m \eta^2 \epsilon^{-2})$.
\end{theorem}
\begin{proof}
  The FOOL algorithm goes through $\log_2 T$ levels of recursion. At each level $h$ of the recursion there are $T/h$ parallel invocations of the REDUCE subroutine with error parameter $\epsilon_h = \Theta(\epsilon h/T)$, which are applied to distributions going from $t$ to $t+h$. By Theorem~\ref{thm-reduce}, each distribution on interval $t, t+h$ has size $E_{t,h} \leq O(m \eta^2 \epsilon_h^{-2}) = O(m \eta^2 \epsilon^{-2} T^2 h^{-2})$. The final distribution has support $E_{1, T} = O(\eta^2 m \epsilon^{-2})$.

Thus, each invocation of the REDUCE subroutine at level $t,h$ takes two distributions of size $E_{t,h/2}$ and $E_{t+h/2,h}$ respectively and returns a new distribution of size $E_{t,h}$. It runs in time $\tilde O(\log n)$, and uses $m n^{o(1)} (E_{t,h} \eta^2 + (E_{t,h/2} + E_{t+h/2,h/2}) \eta h) = O(\eta^3 m^2 T^2 \epsilon^{-2} h^{-2} (\eta + h))$ processors.

The total processor complexity is determined by taking the maximum over all $h$:
\begin{align*}
\text{Processors} &\leq \max_h T/h \times  O(\eta^3 m^2 T^2 \epsilon^{-2} h^{-2} (\eta + h)) \leq O(\eta^4 m^2 T^3 \epsilon^{-2})
\end{align*}
\end{proof}

\subsection{Equivalent transitions}
In \cite{mrs}, Mahajan et al. introduced an additional powerful optimization to reduce the processor complexity. For many types of automata, there are \emph{equivalent transitions} over a given time period $[t, t+h]$, as we define next.
\begin{definition}
For a fixed automaton $i$, starting time $t$ and transition length $h$, we say that $(s_1, s_2) \sim (s_1', s_2')$ if for all $r \in \{0, 1\}^h$ we have
$$
F^h(i, r, s_1, t) = s_2 \Leftrightarrow F^h(i, r, s_1', t) = s_2'
$$
\end{definition}

Recall that Theorem~\ref{thm-reduce} is proved by taking a union-bound over all automata $i$ and states $s, s'$. The key optimization is that it is only necessary to take this union bound over the equivalence classes under $\sim$; this can greatly reduce the support of the distribution.

The algorithms in \cite{mrs} concerned only the limited class of automata  based on counter accumulators. Our algorithm will need to deal concretely with more general types of transitions. Concretely, we assume that for any $t,h$ we list and process canonical representatives for these equivalence classes, which we will denote $\mathcal C_{i,t,h}$. 

\begin{definition}[Canonical transition set]
\label{Bfcanonical-defn}
For an automaton $i$, start time $t$, and a transition period $h$, we define a set $\mathcal C_{i, t,h} \subseteq A \times A$ to be a \emph{canonical transition set} if for every pair $(s_1, s_2) \in A \times A$, there is some pair $(s_1', s_2' ) \in \mathcal C_{i, t,h}$ (not necessarily unique) with $( s_1, s_2 ) \sim ( s_1', s_2' )$.
\end{definition}

There will also be some computational tasks we need to perform on $\mathcal C_{i,t,h}$. First, we define the \emph{canonical stepping table} for $r$ to be the table
$$
\tilde S_{i,t,h} (r)= \{ ( s_1, s_2 ) \in \mathcal C_{i,t,h} \mid F^h(i, r, s_1, t) = s_2 \}
$$

Also, instead of keeping track of $S_{i,t,h}$ during our algorithm, we keep track of only the smaller set $\tilde S_{i,t,h}$.  As in the previous algorithm, we suppose throughout that whenever we form a distribution $D$, we also compute $\tilde S_{i,t,h}(r)$ for every $r \in D$ (abbreviated $\tilde S_{i,t,h}(D)$).

We first show two simple results on the behavior of transitions with respect to $\sim$.
\begin{proposition}
\label{Bunique-prop}
For any strings $r, r'$ with $\tilde S_{i, t,h}(r) = \tilde S_{i,t,h}(r')$, and any state $s_1 \in A$, we must have $F^h(i,r, s_1,t) = F^h(i,r', s_1, t)$.
\end{proposition}
\begin{proof}
Suppose that  $F^h(i,r,s_1,t) = s_2, F^h(i,r',s_1, t) = s_2'$ for $s_2 \neq s_2'$. Then $( s_1, s_2 ) \sim ( v_1, v_2 ) \in \mathcal C_{i,t,h}$, and $( s_1, s_2' ) \sim ( v_1', v_2' ) \in \mathcal C_{i,t,h}$. So any string $w$ has $F^h(i,w,s_1,t) = s_2$ iff $F^h(i,r,v_1, t) = v_2$ and $F^h(i,w,s_1,t) = s_2'$ iff $F^h(i,r,v_1', t) = v_2'$ by definition of $\sim$.  So $( v_1, v_2 ) \in \tilde S_{i,t,h}(r) - \tilde S_{i,t,h}(r')$ and $( v_1', v_2' ) \in \tilde S_{i,t,h}(r') - \tilde S_{i,t,h}(r)$, a contradiction. 
\end{proof}

\begin{proposition}
\label{Bstep-prop1}
For any state $s \in A$, there are at most $|\mathcal C_{i,t,h}|$ states $s'$ of the form $s' = F^h(i,r,s,t)$.
\end{proposition}
\begin{proof}
Let $s'_1, s'_2$ be any two such states. Then $( s, s'_1 ) \not \sim  ( s, s'_2 )$. Hence all such states $( s, s' )$ correspond to distinct equivalence classes under $\sim$. But the total number of such equivalence classes is at most $| \mathcal C_{i,t,h} |$.
\end{proof}

We assume henceforth that we have fixed a set $\mathcal C_{i,t,h}$ for each $i, t, h$. We define a parameter $M_{h}$, which is roughly speaking an upper bound on the computational complexity of working with $\mathcal C_{i,t,h}$. 

\begin{definition}[Efficient canonical transition set]
\label{Befficient-defn}
Let $M_h$ be an increasing real-valued function of the parameter $h$. We say that a canonical transition set $\mathcal C_{i,t,h}$ is \emph{efficient} with cost parameter $M_h$ if $|\mathcal C_{i,t,h}| \leq M_h$ and the following computational tasks can be accomplished in quasi-complexity $(\log n, M_h)$:
\begin{enumerate}
\item[(C1)] We can enumerate $\mathcal C_{i, t,h}$.
\item[(C2)] Given any states $s_1, s_2 \in A$, we can find a pair $( s_1', s_2' ) \in \mathcal C_{i,t,h}$ with $( s_1', s_2' ) \sim ( s_1, s_2 )$.
\item[(C3)] Given $\tilde S_{i,t,h}(r)$ and $s \in A$, then we can determine $F^h(i, r, s, t)$ (by Proposition~\ref{Bunique-prop}, it is uniquely determined.)
\end{enumerate}
\end{definition}

If we define $\mathcal C_{i,t,h} = A \times A$, (i.e. the set of all transitions), then one can easily check that this satisfies the definitions and gives an efficient canonical transition set with $M_h = \eta^2$.

As an example, consider a \emph{counter automaton}, which is defined by maintaining a counter $C$ where on input of bit $x_i$, the counter value is changed by $C \leftarrow C + f(i, x_i)$ for some function $f$. This allows a much smaller canonical transition set in this setting.
\begin{proposition}
For a counter automaton on state space $A = [\eta]$, the set
$$
\mathcal C_{i,t,h} = \bigl \{ (0, z)  \mid z \in [\eta] \bigr \} \cup \bigl \{ (z, 0)  \mid z \in [\eta] \bigr \}
$$
defines an efficient canonical transition set with $M_{h} = O(\eta)$.
\end{proposition}
\begin{proof}
Property (C1) is clear. For Property (C2), note that if $s_2 \geq s_1$ then $( s_1, s_2 ) \sim ( 0, s_2 - s_1 )$, and if $s_2 \leq s_1$ then $(s_1, s_2) \sim (s_1 - s_2, 0)$; so we can accomplish task (C2) with quasi-complexity $(1, 1)$. For Property (C3), note that for any $r$, the table $\tilde S_{i,t,h}(r)$ consists of either a single entry $( 0, x )$, so that given any state $s$ we can determine $F^h(i,r,s,t) = s + x$, or a single entry $(x,0)$ so that given state $s$ we have $F^h(i,r,s,t) = s-x$; this again has quasi-complexity $(1,1)$.
\end{proof}

As the example of counter automata shows, Definition~\ref{Befficient-defn} has a lot of slack in the complexity requirements; 
for many types of automata, these tasks can be accomplished with significantly fewer processors than is required by this definition.

The key result is that we can modify the REDUCE subroutine to take advantage of canonical transitions:
\begin{theorem}
\label{thm-reduce-2}
Let $D_1, D_2$ be distributions on the subintervals $[t, t+h]$ and $[t+h, t + 2h]$ respectively, each of size $n^{O(1)}$.  Then there is an algorithm $\text{REDUCE}(D_1, D_2, \epsilon)$ to construct a distribution $D$ of size $O(\epsilon^{-2} m M_h)$ for the interval $[t, t+2 h]$ with $\text{Err}_{t, 2 h}( D,  D_1 \times  D_2) \leq \epsilon$. It has quasi-complexity $(\log n, m^2 M_{2h}^3 \epsilon^{-2} + m M_{2h}^2 (| D_1| + | D_2| ))$.
\end{theorem}
\begin{proof}
  Let us define $M = M_{2h}$ and $E =2 \epsilon^{-2} m M$. We follow the proof of Theorem~\ref{thm-reduce}, with a few key modifications. First, instead of taking a Chebyshev bound for each transition $i, s, s'$, we only need to take a separate bound for each equivalency class. This accounts for the formula for $E$. Similarly, when we apply Theorem~\ref{discrepancy-thm}, the processor cost is $C_2 + E m M$ where $C_2$ is the processor complexity of the PEO.

Let us discuss how to implement the PEO as used by Theorem~\ref{thm-reduce}. To simplify the discussion,  let us fix an automaton $i$ and discuss only the case when the low bits of $X$ vary.

For any state $s_2 \in A$, let us define the set $Q(s_2) \subseteq A$ to be the set of all states $s_3$ of the form $s_3 = F^{h} (i,r, s_2, t+h)$, for $r_2 \in D_2$. By Proposition~\ref{Bstep-prop1}, we have $|Q(s_2)| \leq |\mathcal C_{i,t,h}| \leq M$; for any fixed $s_2$, Property (C3) allows us to generate the set $Q(s_2)$ with quasi-complexity $(\log n, |D_2| M)$.

Our first step is to build a table $S'_{i,t,h}$ that associates to each $( s_1, s_2 ) \in \mathcal C_{i,t,h}$ and each state $s_3 \in Q(s_2)$ a corresponding $( s_1', s_3' ) \in \mathcal C_{i,t,h}$ with $( s_1', s_3' ) \sim ( s_1, s_3 )$ and a corresponding $( s_2'', s_3'' )$ with $( s_2'', s_3'' ) \sim ( s_2, s_3 )$. To generate this table, we loop over $(s_1, s_2) \in \mathcal C_{i,t,h}$ and form the set $Q(s_2)$. For each $s_3 \in Q(s_2)$, we use Property (C2) to determine $( s_1, s_3 ) \sim ( s_1', s_3' ) \in \mathcal C_{i,t,h}$ and also determine $( s_2, s_3 ) \sim ( s_2'', s_3''  )\in \mathcal C_{i,t+h,h}$. Over all values $(s_1, s_2)$, building $S'$ has quasi-complexity $(\log n, m (|D_2| M^2 + M^3))$.

Having built this table, our PEO is implemented as follows:
\begin{enumerate}
\item Using quasi-complexity $(\log n, |D_2| M)$,  enumerate over $(s_2, s_3) \in \mathcal C_{i+t+h, h}, r_2 \in D_2$ to compute the table $P_1$ defined by  
$$
P1( ( s_2, s_3 ) ) = | \{r_2 \in D_2: F^{h} (i, r_2, s_2, t+h) = s_3 \} |
$$

\item Loop over each $( s_1, s_2 ) \in \mathcal C_{i,t,h}$ and each $s_3 \in Q(s_2)$. Using $S'_{i,t,h}$, find some $( s_2'', s_3'' ) \sim ( s_2, s_3 )$ with $( s_2'', s_3'' ) \in \mathcal C_{i,t+h,h}$.  Sum over $( s_1, s_2 )$ to accumulate the table $P2$ defined by
$$
P2( ( s_1, s_2 ), s_3 ) = \sum_{ \substack{( s_2'', s_3'' ) \sim ( s_2, s_3 )\\( s_2'', s_3'' ) \in \mathcal C_{i,t+h,h}}} P1( ( s_2'', s_3'' ))
$$

By Proposition~\ref{Bstep-prop1}, the total number of such $s_1, s_2, s_3$ is at most $M^2$ and hence this step has complexity $(\log n, M^2)$.

\item Loop over $r_1 \in  D_1$ and pairs $( s_1, s_2 ) \in \tilde S_{i,t,h}(r_1)$ and $s_3 \in Q(s_2)$, and use the table $S'_{i,t,h}$ to find $( s_1', s_3' ) \in \mathcal C_{i,t,h}$ with $( s_1', s_3' ) \sim ( s_1, s_3 )$. Sum these counts over $s_3, (s_1, s_2)$ to obtain the table $P3$ defined by
  $$P3(r_1, ( s_1',  s_3' )) = \sum_{s_3} \sum_{ \substack{( s_1, s_2 ) \in \tilde S_{i,t,h}(r_1) \\ ( s_1', s_3' ) \sim ( s_1, s_3 )}} P2(( s_1, s_2 ), s_3)
$$
By Proposition~\ref{Bstep-prop1}, there are at most $M^2$ tuples $s_1, s_2, s_3$ for each $r_1 \in D_1$. The total quasi-complexity for this step is $(\log n, | D_1 | M^2)$.

\item Finally, sum the table $P3$ over the varying bottom bits of $X$.  
\end{enumerate}
After applying Theorem~\ref{discrepancy-thm}, we have generated the distribution $D$.  To finish, we need to compute $\tilde S_{i,t,2h}(D)$. Consider some $r \in D$ of the form $r = r_1 \# r_2$ for $r_1 \in D_1, r_2 \in D_2$. For each $( s_1, s_2 ) \in \tilde S_{i,t,h} (r_1)$ we use Property (C3) to determine $s_3 = F^{h}(i,r_2,s_2, t+h)$. We then use $S'_{i,t,h}$ to determine $( s_1', s_3' ) \in \mathcal C_{i,t,h}$ with $( s_1', s_3' ) \sim ( s_1, s_3 )$. Finally, store the pair $( s_1', s_3' )$ into $\tilde S_{i,t,h}$. Overall, this step has quasi-complexity $(\log n, E M^2)$. 
\end{proof}

This leads to the following bounds for FOOL:
\begin{theorem}
\label{thm-fool-2}
Suppose we have efficient canonical transition sets with a cost parameter $M_{h}$. Then the FOOL algorithm has quasi-complexity $(\log T \log n, \epsilon^{-2} T^3 m^2 \max_{h} (M_{h}/h)^3 )$. The resulting probability distribution $\tilde D$ has support $|\tilde D| = O(m M_T \epsilon^{-2})$.
\end{theorem}
\begin{proof}
  (This is similar to Theorem~\ref{thm-fool}, so we only provide a sketch). At each level $h$ of the recursion, we invoke $T/h$ parallel instances of REDUCE with error $\epsilon_h = \Theta(\epsilon h/T)$.   Each invocation of REDUCE gives a distribution on the interval $t, h$ of size $E_{t,h} = O(\epsilon_h^{-2} m M_h) = O(m M_h \epsilon^{-2} T^2 / h^2)$. So the final distribution has support $E_{1,T} = O(m M_T \epsilon^{-2})$.

  By Theorem~\ref{thm-reduce-2}, each invocation of REDUCE at stage $t,h$ has quasi-complexity $(\log n, m^2 M_{2 h}^3 \epsilon_h^{-2} + m M_{2 h}^2 ( m M_h \epsilon^{-2} T^2/h^2))$. The processor complexity can be simplified as $m^2 M_{2h}^3 \epsilon^{-2} T^2/h^2$. As there are $T/h$ parallel applications of REDUCE at level $h$, the overall processor count at level $h$ is at most $\epsilon^{-2} T^3 m^2 h^{-3} M_h^3$.
\end{proof}

\section{Applications of fooling automata}
\label{app-fool}

\subsection{Set discrepancy}
\label{set-discrep-sec}
To illustrate how to use automata-fooling as a building-block for NC algorithms, consider the problem of \emph{set discrepancy}. A simple example, we show how to to match Chernoff bound discrepancy for linear functionals with integer-valued coefficients and variables.

\begin{proposition}
  Given $m$ linear functionals in $n$ variables with coefficients $\{-1, 0, 1 \}$, there is an algorithm to find $X_1, \dots, X_n \in \{-1, 1 \}^n$ such that
$$
|l_j(X)| \leq O(\sqrt{n \log m}) \qquad \text{for all $j = 1, \dots, m$}
$$ 
using quasi-complexity $(\log^2 mn, m^4 n^3)$.
\end{proposition}
\begin{proof}
For each $j$, we represent the ``statistical test'' $|l_j(X)| \leq O(\sqrt{n \log m})$ by an automaton which receives the input values $X_1, \dots, X_n$ in order and maintains a running counter $\sum_{k=1}^t l_{jk} X_k$. When $X_i$ are independent then with high probability each of the automata terminates in a state with counter bounded by $O(\sqrt{n \log m})$.  These automata run for $T =n$ timesteps.

Now consider a distribution $D$ which fools these automata to error $\epsilon = 1/(2 m)$; i.e. its transition probabilities are close to what one would obtain from the independent distribution. Then with positive probability, these automata run on $D$ will have discrepancy $O(\sqrt{n \log m})$. In particular, we can spend $|D| m n$ processors to test all possible elements of the distribution and produce a desired vector $\vec X$.

Let us first give a simple estimate using Theorem~\ref{thm-fool}, and then show how to optimize this by more careful counting of states and equivalence classes of states. First, as the state space has size $\eta = O(n)$ (to count the running sum), Theorem~\ref{thm-fool} gives a distribution to fool the automata in quasi-complexity $(\log^2 n, \eta^4 m^2 T^3 \epsilon^{-2}) = (\log^2 n, m^4 n^7)$. The resulting distribution has support $m^3 n^2$, thus the final testing step requires just $m^4 n^3$ processors, which is negligible. Indeed, this final testing step is almost always negligible for automata-fooling.

To improve this, note that in a window of size $h$, the only thing that matters is the total change in the value of the running sum. The starting state does not matter, and the running sum can only change by up to $h$; thus we can take a set of equivalent transitions of cost $M_{h} = h$. So by  Theorem~\ref{thm-fool-2} the total processor complexity is at most $(m n)^{o(1)} \epsilon^{-2} T^3  m^2 \epsilon^{-2} \max_h (M_h/h)^3 \leq m^{4+o(1)} n^{3+o(1)}$.
\end{proof}
By contrast, \cite{mrs} required roughly $m^{10} n^7$ processors for this task.

\subsection{The Gale-Berlekamp Game}
The method of fooling automata can give better constants for the Gale-Berlekamp Game compared to the Fourth-Moment Method, at the cost of greater time and processor complexity.
\begin{theorem}
There is an algorithm with quasi-complexity $(\log^2 n, n^5)$ to find $x, y \in \{-1, +1\}^n$ satisfying $ \sum_{i,j} a_{i,j} x_i y_j \geq (\sqrt{2/\pi} - o(1)) n^{3/2}$.
\end{theorem}
\begin{proof}
As in Theorem~\ref{gb-game1}, we set $R_i(y) = \sum_j a_{ij} y_j$ and $x_i =(-1)^{[R_i(y) < 0]}$, so that we wish to maximize $\sum_i |R_i(y)|$. As shown in \cite{brown,spencer}, for $y \sim \{0, 1\}^n$ we have $\bE[|R_i(y)|] \geq \sqrt{2 n/\pi}$.  Each term $R_i(y)$ can be computed from the bits $y$ by an automaton which counts the running sum. This automaton can be implemented using $M_h = O(h)$. Thus, we can build a distribution $D$ which fools these $n$ automata to error $\epsilon = o(1)$; for $y \sim D$ we have $\bE[ |R_i(y)| ] \geq \sqrt{2 n/\pi - o(1)}$. By Theorem~\ref{thm-fool-2}, the distribution $D$ can be generated using quasi-complexity $(\log^2 n, \epsilon^{-2} T^3 m^2)$; the processor count can be simplified as $n^{5+o(1)}$.
\end{proof}

\subsection{Johnson-Lindenstrauss Lemma}
In another application of \cite{sivakumar}, consider the well-known dimension-reduction procedure based on the Johnson-Lindenstrauss Lemma \cite{jl}. (We present it in a slightly modified form, based on approximating norms instead of distances.) We are given a collection of $n$ unit vectors $U \subseteq \mathbf R^d$, and we wish to find a linear projection $\phi: \mathbf R^d \rightarrow \mathbf R^k$ for $k = O(\delta^{-2} \log n)$ such that every $u \in U$ satisfies
\begin{equation}
\label{jl-eqn}
(1 - \delta) \leq || \phi(u)  || \leq (1+\delta)
\end{equation}
where here $|| \cdot ||$ denotes the euclidean $\ell_2$ norm.

The original construction was randomized. In \cite{engebretsen}, a sequential derandomization was given using the method of conditional expectations. There has since been significant research in developing sequential deterministic algorithms for a variety of settings, see e.g. \cite{dadush, kane}.  

Achlioptas \cite{achlioptas} gives a simple randomized algorithm where $\phi$ is a random $k \times d$ matrix $L$, whose entries are $+1$ or $-1$ with equal probability. As shown in \cite{achlioptas}, this satisfies the condition (\ref{jl-eqn}) with high probability. Interpreting this as a type of finite automaton, Sivakumar \cite{sivakumar} gives an NC derandomization. We make a number of modifications to optimize Sivakumar's algorithm.

The randomized construction of Achlioptas can be easily transformed into a log-space statistical test. Here, the driving bits are the entries of the matrix $L$. For each $u \in U$, we need to compute $|| \phi(u) ||$. To do so, for each $i = 1, \dots, k$ we compute the sum $s_i = \sum_j L_{ij} u_j$. We also compute the running sum of sums $r = \sum_i s_i^2$. If the final value $r$ satisfies $r \leq \delta^2$ we output SUCCESS. As shown in \cite{achlioptas}, the expected number of values $u$ which fail this test is at most $1/2$, for appropriately chosen $k$.

\begin{theorem}
\label{jl-thm}
There is an algorithm to find  $\phi$ satisfying condition (\ref{jl-eqn}) for every $u \in U$, where $k = O(\delta^{-2} \log n)$, using quasi-complexity $(\log^2 n, n^4 d^5 \delta^{-12} + n^{4} \delta^{-16})$.
\end{theorem}
\begin{proof}
We begin with a few simple observations. By rescaling $\delta$, it suffices to achieve relative error $O(\delta)$. Also, we may assume $d, \delta$ are polynomially bounded in $n$. For, if $d \geq n$, then we can perform a change of basis so that $v_1, \dots, v_n$ are unit vectors, and the remaining $n - d$ coordinates are zero. If $\delta \leq n^{-1/2}$, then we can simply take $k = n$ and $\phi$ is the identity map. 

For any $u$, the automaton we have described computes separate summations $s_1, \dots, s_k$; these all have the same distribution, and depend on independent random bits. If we apply the automaton-fooling theorem directly, we are creating separate distributions to fool each of these summations; but this is redundant, as a single distribution will fool any of them. We can instead create a single distribution to fool the automaton computing $s_1$, and replicate it $k$ times.

Thus, for each $u$ we define the automaton $R_1(u)$, which receives $d$ driving values $y_1, \dots, y_d \in \{-1, 1\}^d$, to compute the running sum $s = \sum_j y_j u_j$. We begin by applying Theorem~\ref{thm-fool-2} to fool these $n$ automata.

Next, let $ D_1$ be a distribution fooling $R_1$. We want to create a distribution which fools the sum $s_1 + \dots + s_k$ to error $\epsilon = 1/n$. We do so recursively. First, we create an automaton $R_2$ which takes two independent draws from $ D_1$ and computes $s_1 + s_2$. We can fool $R_2$, giving us a distribution $ D_2$ which fools the sum $s_1 + s_2$. It also automatically fools the sums $s_3 + s_4, s_5 + s_6,  \dots$ as these have the same distribution.  Proceeding in this way, we recursively set $D_{i+1} = \text{REDUCE}( D_i,  D_i, \epsilon 2^i/k)$ to fool automaton $R_{2^i}$. After $\log_2 k$ steps, the distribution $ D_{\text{final}} =  D_{\log_2 k}$ fools the overall sum $s_1 + \dots + s_k$ to error $\epsilon$. This recursive application of REDUCE to larger and larger transitions, is very similar to the recursion used to define FOOL itself. The one key difference here is that at each level of this recursion, we only perform a single application of REDUCE

\textbf{Fooling automaton $R_1$:} We want to fool $R_1$ to error $n^{-1} k^{-1}$. In order to achieve a finite-state automaton, we the running sum $s$ as well as the vector $u$ to integer multiples of a parameter $x$. We first note that this sum $s$ can be maintained within a narrow window. For, while it is possible that that sum $s$ achieves (either at the end or in an internal state) a value as large as $\sum_j |v_j|$, this is unlikely. By Bennet's inequality coupled with the fact that $\sum_j v_j^2 = 1$, it holds with high probability (over random inputs) that $|s| \leq O( \frac{\log n}{\log \log n})$. So, we maintain $s$ within this window; it it leaves this window, then we can force automaton $R_1$ to go into a FAIL state.

The quantized sum $\hat s$ we compute could differ from the true $s$ by a factor of $d x$. This can causes an error of at most $|s| d x \leq n^{o(1)} d x$ in the term $s^2$, which could cause an error of $n^{o(1)} k d x$ in the final sum $r$ (which is the sum of $k$ copies of $s$). As we are only trying to ensure that $r \leq O(\delta^2)$, it suffices to take
$$
x = \delta^2/(d k)
$$

Let us count the number of inequivalent transitions in an interval $[t, t+h]$ for the running sum computation. The starting state is irrelevant, and within this window the sum $s$ changes by at most $\sum_{j=t}^{t + h} |v_j| \leq \sqrt{h}$. This corresponds to $\sqrt{h}/x$ distinct states after quantization, and so
$$
M_{h} \leq d k \sqrt{h} \delta^{-2}
$$

We now apply Theorem~\ref{thm-fool-2} to construct a distribution $ D$ fooling these $n$ automata to error $\epsilon = n^{-1} k^{-1}$. The number of driving bits is $T = d$ so this requires quasi-complexity $(\log^2 n, d^5 n^4 \delta^{-12})$.

\textbf{Fooling $R_{2^i}$, for $i \geq 1$:} We do this recursively through $i = 1, \dots, \log_2 k$ stages; at each step, we consider a single REDUCE step on the simple automaton $R_{2^i}$ which takes two draws from $ D_{i-1}$, uses these to compute respectively the sums $s_1 + \dots +s_{2^{i-1}}$ and $s_{2^{i-1} + 1} + \dots + s_{2^i}$, and adds these quantities together.

The running sum $r$ can be quantized here to multiples of $k^{-1} \delta^{2}$, to ensure a final truncation error of size at most $O(\delta^2)$. As we have restricted the automaton $R_1$ to maintain the running sum within an $n^{o(1)}$ window, this implies that automaton $R_{2^{i-1}}$ only needs to keep track of $2^i \delta^{-2} k n^{o(1)}$ states. Thus by Theorem~\ref{thm-reduce-2}, the distribution $|  D_{i}|$ can be taken to have support
$$
|  D_i | = n^{o(1)} 2^i n^2 \delta^{-8} (\epsilon^{-1} k / 2^i)^2 
$$

Similarly, when combining these two automata, the number of relevant state transitions is at most $2^i k \delta^{-2} n^{o(1)}$. (The sum $r$ can change by $n^{o(1)} 2^i$ in this time horizon). Thus by Theorem~\ref{thm-reduce-2}, the work for a single REDUCE step at iteration $i$ is $\tilde O(\log n)$ time and 
$$
n^{o(1)} 2^i n^2 \delta^{-8} (\epsilon^{-2} k / 2^i)^2 \times 2^i k \delta^{-2} n^{o(1)} = n^{4 + o(1)} \delta^{-16}
$$
processors.
\end{proof}

This version of the Johnson-Lindenstrauss Lemma can be used to solve the (more standardly stated) version based on fooling distances:
\begin{proposition}
Given a set $U \subseteq \mathbf R^d$ of vectors with $b$ bits of precision and size $|U| = n$, there is an algorithm to find a linear mapping $\phi : \mathbf R^d \rightarrow R^k$ for $k = O(\delta^{-2} \log n)$ such that
$$
(1-\delta) || u - u' || \leq || \phi(u) - \phi(u') || \leq (1+\delta) || u - u'||
$$
using quasi-complexity $(\log^2 n + \log b, n^8 d^5 \delta^{-12} + n^8 \delta^{-16} + n^2 b)$.
\end{proposition}
\begin{proof}
First, form the set of unit vectors $U' = \{ \frac{u - u'}{||u - u'||} \mid u, u' \in U \}$. This requires quasi-complexity $(\log nb, n^2 b)$. Next apply Theorem~\ref{jl-thm} to the set $U'$ which has cardinality $|U'| \leq n^2$.
\end{proof}

\section{Acknowledgments}
Thanks to Aravind Srinivasan, for extensive comments and discussion. Thanks to anonymous journal referees for careful review and a number of helpful helpful suggestions.

\appendix

\end{document}